\documentclass[numsec,webpdf,modern,medium]{oup-authoring-template}
\onecolumn % for one column layouts

\usepackage{amsmath,amssymb,amsthm}
\theoremstyle{thmstyleone}%
\newtheorem{theorem}{Theorem}%  meant for continuous numbers
\newtheorem{proposition}[theorem]{Proposition}%
\theoremstyle{thmstyletwo}%
\newtheorem{remark}{Remark}%
\theoremstyle{thmstylethree}%
\newtheorem{definition}{Definition}
\usepackage{algpseudocode}
\usepackage{iftex}
\usepackage{graphicx}
\usepackage{booktabs}
\usepackage{threeparttable}
\usepackage{longtable}
\usepackage{array}
\usepackage{verbatim}
\begin{document}

\journaltitle{Journal Title Here}
\DOI{DOI added during production}
\copyrightyear{YEAR}
\pubyear{YEAR}
\vol{XX}
\issue{x}
\access{Published: Date added during production}
\appnotes{Paper}

%\firstpage{1}

%\subtitle{Subject Section}

\title[Partial WaveCanCoh]{Partial Wavelet Canonical Coherence for Nonstationary Signals with High Dimensional Confounders}

\author[1,$\ast$]{Haibo Wu}
\author[2]{Marina I. Knight}
\author[1]{Hernando Ombao}
% \author[3]{Fourth Author}
% \author[4]{Fifth Author\ORCID{0000-0000-0000-0000}}
\address[1]{\orgdiv{Statistics Program}, \orgname{King Abdullah University of Science and Technology}, \orgaddress{\street{Street}, \postcode{23955}, \state{Thuwal}, \country{Saudi Arabia}}}
\address[2]{\orgdiv{Department of Mathematics}, \orgname{University of York}, \postcode{YO10 5DD}, \state{York}, \country{United Kingdom}}

% \address[3]{\orgdiv{Department}, \orgname{Organization}, \orgaddress{\street{Street}, \postcode{Postcode}, \state{State}, \country{Country}}}
% \address[4]{\orgdiv{Department}, \orgname{Organization}, \orgaddress{\street{Street}, \postcode{Postcode}, \state{State}, \country{Country}}}

\corresp[$\ast$]{Corresponding author: \href{haibo.wu@kaust.edu.sa}{haibo.wu@kaust.edu.sa}}

% \received{Date}{0}{Year}
% \revised{Date}{0}{Year}
% \accepted{Date}{0}{Year}

%\editor{Associate Editor: Name}

%\abstract{
%\textbf{Motivation:} .\\
%\textbf{Results:} .\\
%\textbf{Availability:} .\\
%\textbf{Contact:} \href{name@email.com}{name@email.com}\\
%\textbf{Supplementary information:} Supplementary data are available at \textit{Journal Name}
%online.}

\abstract{We develop Partial Wavelet Canonical Coherence for measuring the direct canonical association between two multivariate nonstationary time series after adjustment for possibly high-dimensional confounders. To the best of our knowledge, this is the first method that establishes a frequency-domain formulation of the partial canonical correlation analysis for time series. Through a wavelet approach, the proposed method yields a scale-specific, time-varying measure of association capable to work with potential data nonstationarities. We formulate the target quantity under the multivariate locally stationary wavelet framework, develop principled estimation through local wavelet spectral matrices, and incorporate principal-component reduction for stable adjustment in high-dimensions. Simulations show that the method removes spurious marginal association induced by confounding and accurately recovers direct association, including in higher-dimensional settings. Analysis of U.S. exchange-traded funds reveals substantial time-varying and scale-dependent direct canonical association after adjustment for external market effects.} 


% \keywords[Abbreviations]{abbreviation1, abbreviation2, abbreviation3, abbreviation4}

% \otherabstract[Additional Abstract]{Use this element for elements such as Graphical abstract, Lay summary, Translated abstract etc. Que cum aut etum qui ium dolupta ssequia autati odis demporepe ad et es alit rem repudaerae min et volorum re volupta nobit volectur aut fuga.}

% \otherabstract[Graphical Abstract]{\colorbox{black!20}{\hbox to 0.97\textwidth{\vbox to 50pt{}}}}

% \boxedtext{Key Messages}{
% \begin{itemize}
% \item Key boxed text here.
% \item Key boxed text here.
% \item Key boxed text here.
% \end{itemize}}

\maketitle

%\begin{epigraph}
%Epigraph text. Ximporem qui reperov idempedit modio. Bisto imagnatem quae aceptis
%nobitae quid eum rae adignis quias-sit vellacc uptatur sunt quis rentis eaquasit alia deliquam
%rec-to consed unt. Empor sum ratur ressimusdae. Nam fugiae.
%\source{Epigraph source}
%\end{epigraph}

\section{Introduction}
\label{sec:intro}
In empirical financial studies, an important objective is to assess the direct association between two groups of assets, rather than between two individual assets, after accounting for external market influences and common risk factors. This question is particularly relevant when examining whether a group of broad market, growth, small-cap, and technology-oriented exchange-traded funds (ETFs) is directly associated with a group of sector ETFs representing financially and macroeconomically sensitive segments of the U.S. equity market. 
The type of dependence that should be addressed in financial time series is between groups of assets group rather than a single pair of assets. Notably, the observed association between two asset groups may be induced or amplified by external market movements, including international equity fluctuations, country-level shocks, or other common risk factors. Without proper adjustment, marginal cross-group dependence may therefore overstate the direct linkage between the two groups. Moreover, this direct dependence is likely to be dynamic and frequency-specific: it may strengthen during periods of market stress, weaken in stable regimes, and vary across investment horizons. Short-run co-movements may reflect liquidity effects or market-wide rebalancing, whereas lower-frequency components may be more closely related to macroeconomic or global market conditions. These considerations motivate the need for a framework that can quantify direct multivariate dependence between asset groups in a time- and frequency-localized manner; see Figure~\ref{fig:fig1_bigpicture}.

Existing methods do not provide a unified solution to this problem. Time-domain correlation and canonical correlation methods can summarize multivariate association, but they cannot identify the frequency-specific components underlying the dependence. Classical coherence methods provide frequency-domain information, yet they are typically pairwise or unadjusted, making it difficult to separate direct cross-block dependence from association induced by external market factors. Partial canonical correlation analysis allows adjustment for confounders, but existing extensions to multivariate time series are formulated in the time domain and do not capture the time-varying and frequency-specific structure of nonstationary signals. These limitations leave a methodological gap for direct dependence analysis between two sets of nonstationary multivariate time series. To address this gap, we propose Partial Wavelet Canonical Coherence (Partial WaveCanCoh), a framework for inference on time-varying and scale-specific direct canonical dependence between two multivariate time series blocks after accounting for external confounding processes, also with potential nonstationary character.

\begin{figure}[htbp]
    \centering
    \includegraphics[width=0.75\linewidth]{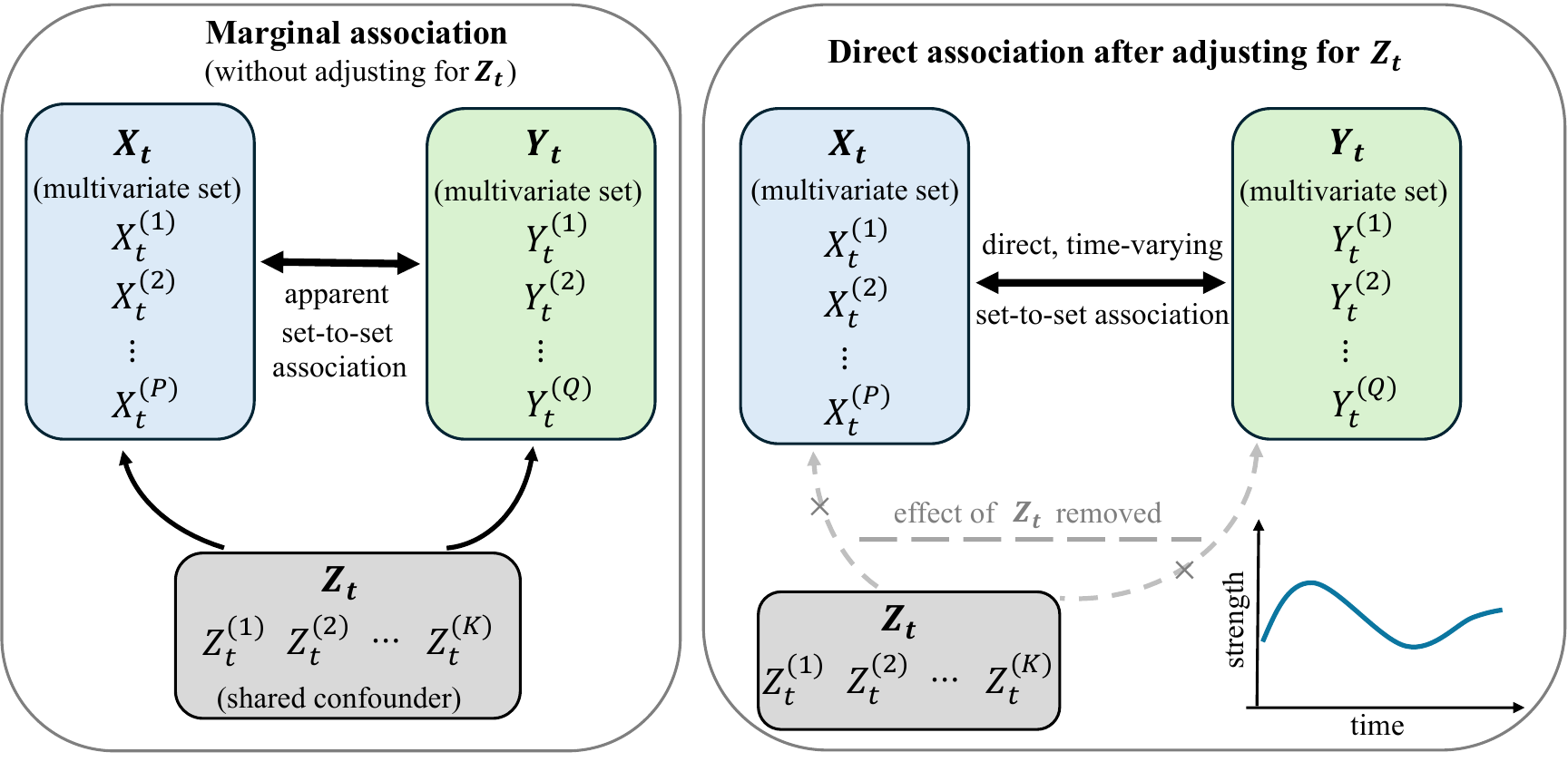}
    \caption{The template: distinguishing direct and confounded time-varying association between two sets of multivariate processes. The left panel shows the marginal association between $\{\mathbf{X}_t\}$ and $\{\mathbf{Y}_t\}$, which may be induced or distorted by a shared (high dimensional) confounding process $\{\mathbf{Z}_t\}$. The right panel shows the target of interest: the direct, time-varying canonical association between $\{\mathbf{X}_t\}$ and $\{\mathbf{Y}_t\}$ after removing the effects from $\{\mathbf{Z}_t\}$.}
    \label{fig:fig1_bigpicture}
\end{figure}

A classical and widely used approach for {\em quantifying dependence between two sets of variables} is canonical correlation analysis (CCA) \citep{hotelling1936relations}. To account for the effect of an additional, potentially confounding set of variables, \cite{rao1969partial} introduced partial canonical correlation. These ideas were subsequently extended to the time series setting: \cite{box1977canonical} and \cite{geweke1982measurement} adapted CCA to characterize dependence between multivariate time series, and \cite{Serge1990partial} extended the partial CCA framework to multivariate time series. In many applications, including finance and neuroscience, dependence is inherently frequency-specific, motivating frequency-domain analogues of such set-to-set association measures. Using Fourier analysis, \cite{brillinger2001time} developed a spectral version of CCA for time series, which quantifies the overall association between two sets of series as a function of frequency. Despite this development, to the best of our knowledge, a corresponding frequency-domain formulation of \emph{partial} CCA has not yet been established, even though this is is conceptually natural. %building on spectral CCA   
%    \textcolor{red}{there is a claim about novelty}

At the same time, many real-world time series exhibit pronounced {\em nonstationarity}, making purely Fourier-based approaches restrictive in practice. Owing to its excellent time-localization properties, wavelet methodology has become a standard tool for analyzing nonstationary time series and studying dependence in multivariate time series \citep{Cohen2011, ombao2014, schulte2016wavelet}. Nevertheless, most existing wavelet-based work focuses on pairwise coherence between individual channels or nodes, rather than on set-to-set measures that summarize association between two collections of time series. Addressing this limitation, \cite{wu2025wavelet} introduced a wavelet-domain framework for quantifying time-varying coherence between two sets of time series based on the multivariate locally stationary wavelet process, providing, the first principled solution to the {\em set-to-set coherence} problem in the wavelet setting. However, this framework does not explicitly adjust for the effects of potentially confounding variables. Related work on partial wavelet coherence has been developed in \cite{ombao2014, ding2025partial}, but these approaches are formulated at the level of two individual channels and do not extend to set-to-set coherence between two multivariate collections.%to our knowledge, 

In this paper, we develop a rigorous wavelet-domain framework for quantifying partial canonical coherence between two sets of multivariate time series while explicitly adjusting for the influence of additional, potentially confounding multivariate processes, all featuring a potential nonstationary character. Our approach builds on the set-to-set wavelet coherence framework of \cite{wu2025wavelet} and provides a principled measure of direct association that is localized in scale and accommodates nonstationary dependence. The main contributions of this work are summarized as follows: 
(1) we introduce, to our knowledge, the first measure of partial canonical correlation for multivariate time series to the wavelet (time–frequency) domain; 
(2) by leveraging the time-localization property of wavelets, the proposed framework naturally accommodates nonstationary dependence and avoids the global stationarity assumptions inherent to Fourier-based formulations; 
(3) we develop a complete set of definitions together with rigorous estimation procedures of the proposed measures; and 
(4) the proposed framework can be integrated with common dimension-reduction techniques for multivariate time series, enabling scalable adjustment when the confounding processes are high-dimensional. In particular, high-dimensional confounding blocks may be summarized using dynamic factor models \citep{lam2012factor, Mario2000DFM}, dynamic principal component methods \citep{pena2016generalized}, or wavelet-domain dimension reduction approaches such as the adaptive wavelet-domain principal component analysis of \cite{knight2024adaptive}. These techniques provide a natural interface between high-dimensional confounding adjustment and our proposed wavelet-domain partial canonical coherence analysis.

The paper is organized as follows. Section \ref{sec:related_methods} provides a brief review of classical partial canonical correlation analysis  and wavelet canonical coherence for multivariate time series. Section \ref{sec:partial_wavecancoh} introduces the proposed partial wavelet canonical coherence (Partial WaveCanCoh) framework, including its formal definition and the associated estimation procedures. Section \ref{sec:high_dim} discusses extensions to high-dimensional settings and presents practical strategies for handling large confounding blocks. Section \ref{sec:sim} reports simulation studies evaluating the finite-sample performance of the proposed methods. Section \ref{sec:data_analysis} presents a real-data application to financial markets. Section \ref{sec:conc} concludes with a discussion and  directions for future research.

\section{Related Methods}\label{sec:related_methods}

\subsection{Partial canonical correlation analysis}\label{sec:partial_cca}
Before introducing the framework of partial canonical coherence for multivariate time series, we first provide a brief overview of classical partial canonical correlation analysis. The work of \cite{rao1969partial} generalizes the notion of partial correlation—originally defined between two random variables $X$ and $Y$ conditional on a third variable $Z$—to the setting of multivariate random vectors. Specifically, it extends this concept to define the partial canonical correlation between two sets of variables, $\mathbf{X}$ and $\mathbf{Y}$, conditioned on a third set, $\mathbf{Z}$. This formulation builds upon the foundational theory of canonical correlation analysis developed by \cite{hotelling1936relations}, by removing the linear influence of $\mathbf{Z}$ from both $\mathbf{X}$ and $\mathbf{Y}$ prior to evaluating their mutual association. 

Suppose there are three sets of variables  $\mathbf{X}=\left(X^{(1)}, \ldots, X^{(P)}\right)^{\top}$, $\mathbf{Y}=\left(Y^{(1)}, \ldots, Y^{(Q)}\right)^{\top}$ and $\mathbf{Z}=\left(Z^{(1)}, \ldots, Z^{(K)}\right)^{\top}$ from a certain joint distribution, and the goal is to quantify the direct association between $\mathbf{X}$ and $\mathbf{Y}$ after removing the linear effects of $\mathbf{Z}$. Specifically, this is achieved by considering the linear combinations of the residuals from regressing $\mathbf{X}$ and $\mathbf{Y}$ on $\mathbf{Z}$, and finding those combinations that yield maximal correlation. Formally, the partial canonical correlation between $\mathbf{X}$ and $\mathbf{Y}$ given $\mathbf{Z}$ is defined as,
\begin{align*}
    \boldsymbol{\rho}(\mathbf{X}, \mathbf{Y} \mid \mathbf{Z}):= \operatorname{max}_{\mathbf{a,b}} \operatorname{Corr}\left(\mathbf{a}^{\top}\mathbf{R}_\mathbf{X}, \mathbf{b}^{\top}\mathbf{R}_\mathbf{Y}\right),
\end{align*}
with residual errors obtained by following multivariate linear regression modeling, namely,
\begin{align*}
    \mathbf{R}_\mathbf{X} &= \mathbf{X} - {\boldsymbol{\beta}}_{\mathbf{X}} \mathbf{Z},  \, \mbox{ where the dimension of }\mathbf{R}_\mathbf{X} \mbox{ is } P \times 1,\\ 
    \mathbf{R}_\mathbf{Y} &= \mathbf{Y} - {\boldsymbol{\beta}}_{\mathbf{Y}}
    \mathbf{Z},  \, \mbox{ the dimension of }  \mathbf{R}_\mathbf{Y} \mbox{ is } Q \times 1,
\end{align*} 
and $\boldsymbol{\beta}_{\mathbf{X}}$ $(P \times K)$ and $\boldsymbol{\beta}_{\mathbf{Y}}$ $(Q \times K)$ are regression coefficients for $\mathbf{X}$ and $\mathbf{Y}$ on $\mathbf{Z}$, respectively. Assume we denote $\boldsymbol{\Sigma}$ as the covariance matrix of the $(P+Q+K)$ variables pertaining to the three sets,
\begin{align}
\boldsymbol{\Sigma}= \left[\begin{array}{lll}
\boldsymbol{\Sigma}_{\mathbf{XX}} & \boldsymbol{\Sigma}_{\mathbf{XY}} & \boldsymbol{\Sigma}_{\mathbf{XZ}} \\
\boldsymbol{\Sigma}_{\mathbf{YX}} & \boldsymbol{\Sigma}_{\mathbf{YY}} & \boldsymbol{\Sigma}_{\mathbf{YZ}} \\
\boldsymbol{\Sigma}_{\mathbf{ZX}} & \boldsymbol{\Sigma}_{\mathbf{ZY}} & \boldsymbol{\Sigma}_{\mathbf{ZZ}}
\end{array} \right]_{(P+Q +K) \times (P+Q+K)}
\end{align}
where $\boldsymbol{\Sigma}_{\mathbf{\cdotp \, \cdotp}}$ denotes variance (covariance) among $\mathbf{X}$, $\mathbf{Y}$, and $\mathbf{Z}$ correspondingly and we assume the means of $\mathbf{X}$, $\mathbf{Y}$, and $\mathbf{Z}$ to be zero. The optimal parameters are %obtained  as
\begin{align*}
    {\boldsymbol{\beta}}_{\mathbf{X}} &=  {\boldsymbol{\Sigma}}_{\mathbf{XZ}}{\boldsymbol{\Sigma}}_{\mathbf{ZZ}}^{-1},\\
    {\boldsymbol{\beta}}_{\mathbf{Y}} &= {\boldsymbol{\Sigma}}_{\mathbf{YZ}}{\boldsymbol{\Sigma}}_{\mathbf{ZZ}}^{-1}.
\end{align*}
The direct overall linear association between two sets of variables $\mathbf{X}$ and $\mathbf{Y}$ conditional on $\mathbf{Z}$  can be measured through the canonical correlation between residuals $\mathbf{R}_\mathbf{X}$ and $\mathbf{R}_\mathbf{Y}$. Denoting by  $\boldsymbol{\Sigma}_{\mathbf{XX \cdot Z}}$,  $\boldsymbol{\Sigma}_{\mathbf{YY \cdot Z}}$ and $\boldsymbol{\Sigma}_{\mathbf{XY \cdot Z}}$ the covariance matrices of the (zero-mean) residual vectors, we have,
\begin{align*}
    \boldsymbol{\Sigma}_{\mathbf{XX \cdot Z}} &= \operatorname{E} [\mathbf{R}_\mathbf{X} \mathbf{R}_\mathbf{X}^{\top} ] = \boldsymbol{\Sigma}_{\mathbf{XX}} - \boldsymbol{\Sigma}_{\mathbf{XZ}}\boldsymbol{\Sigma}_{\mathbf{ZZ}}^{-1} \boldsymbol{\Sigma}_{\mathbf{ZX}},\\ 
    \boldsymbol{\Sigma}_{\mathbf{YY \cdot Z}} &= \operatorname{E}  [\mathbf{R}_\mathbf{Y} \mathbf{R}_\mathbf{Y}^{\top} ] = \boldsymbol{\Sigma}_{\mathbf{YY}} - \boldsymbol{\Sigma}_{\mathbf{YZ}}\boldsymbol{\Sigma}_{\mathbf{ZZ}}^{-1} \boldsymbol{\Sigma}_{\mathbf{ZY}}, \\ 
        \boldsymbol{\Sigma}_{\mathbf{XY \cdot Z}} &= \operatorname{E}  [\mathbf{R}_\mathbf{X} \mathbf{R}_\mathbf{Y}^{\top} ] = \boldsymbol{\Sigma}_{\mathbf{XY}} - \boldsymbol{\Sigma}_{\mathbf{XZ}}\boldsymbol{\Sigma}_{\mathbf{ZZ}}^{-1} \boldsymbol{\Sigma}_{\mathbf{ZY}}.
\end{align*}
Based on the definition of classic canonical correlation between two sets of variables and the above results, the partial canonical correlation coefficient, $\boldsymbol{\rho}_{\mathbf{XY\cdot Z}}$, can therefore be defined as
\begin{align}
    &\boldsymbol{\rho}_{\mathbf{XY\cdot Z}} = \operatorname{max}_{\mathbf{a,b}} \{\mathbf{a}^{\top}  \boldsymbol{\Sigma}_{\mathbf{XY \cdot Z}}\mathbf{b}\}, \, \text{subject to }   \\
&\mathbf{a}^{\top} \boldsymbol{\Sigma}_{\mathbf{XX \cdot Z}}\mathbf{a} = 1 \text{ and } \mathbf{b}^{\top}  \boldsymbol{\Sigma}_{\mathbf{YY \cdot Z}} \mathbf{b} =1, \notag
\end{align}
where $\{\mathbf{a}\}_{P \times 1}$ and $\{\mathbf{b}\}_{Q \times 1}$ are partial canonical vectors for $\mathbf{X}$ and $\mathbf{Y}$, respectively. 

As noted earlier, partial canonical correlation analysis has not been developed for nonstationary multivariate time series, particularly in the frequency domain. Additionally, existing work on multivariate time series has largely focused on the stationary setting in time domain. For example, \cite{Serge1990partial} developed a canonical partial autocorrelation function for stationary multivariate processes based on the canonical analysis of forward and backward innovations, but this does not provide a framework for the partial canonical association between two nonstationary multivariate processes.
We therefore begin by reviewing wavelet canonical coherence (WaveCanCoh, \cite{wu2025wavelet}), which provides the natural starting point for our development. While WaveCanCoh characterizes scale-specific, time-varying canonical association between two sets of nonstationary multivariate processes, it does not adjust for confounding effects; the proposed Partial WaveCanCoh framework is designed to address precisely this limitation.
\subsection{Wavelet canonical coherence (WaveCanCoh)} \label{sec:wavecancoh}

Both WaveCancoh and our proposed Partial WaveCanCoh are developed within the MvLSW framework \citep{nason2000wavelet, ombao2014}, which provides a flexible representation for multivariate nonstationary time series using wavelet bases. For readers unfamiliar with wavelet methods, more details can be found in \cite{daubechies1992ten}.

The locally stationary wavelet (LSW) model was first proposed by \cite{nason2000wavelet} for univariate nonstationary time series  (see \cite{dahlhaus2000likelihood} for the Fourier counterpart and \cite{palasciano2025continuous} for the continuous-time version) and later extended to multivariate processes by \cite{ombao2014}. Consider a $P$-dimensional stochastic process $\{\mathbf{X}_t\}_{t=1}^T$ with $\mathbf{X}_t = \left(X_t^{(1)}, X_t^{(2)}, \ldots, X_t^{(P)}\right)^\top$ at each $t=1,\ldots,T$, whose second-order structure may evolve over time. Often the individual $p$th time series is referred to as a `channel', a terminology we also adopt here, with $p=1, \ldots, P$. 

Under the multivariate locally stationary wavelet (MvLSW) representation, the process is 
\begin{align} \label{equ: MvLSW}
\mathbf{X}_t = \sum_{j=1}^{\infty}\sum_{k\in\mathbb{Z}}
\mathbf{V}_j(k/T)\psi_{j,k}(t)\boldsymbol{\xi}_{j,k},
\end{align}
where $\{\psi_{j,k}\}$ denotes the discrete non-decimated wavelet system indexed by scale $j$ and location $k$; the vectors $\{\boldsymbol{\xi}_{j,k}\}$ are mutually uncorrelated $P$-dimensional random vectors with mean $\mathbf{0}$ and identity covariance matrix; the matrix-valued function $\mathbf{V}_j(u)$ is a $P\times P$ lower-triangular transfer matrix that determines how variation at scale $j$ contributes to the multivariate dependence structure locally at rescaled time $u=k/T$. Since the wavelet functions $\psi_{j,k}(t)$ are localized in both time and frequency, the coefficients $\mathbf{V}_j(u)$ allow the second-order properties of the process to vary across time while remaining scale-dependent. The evolving dependence structure of $\{\mathbf{X}_t\}$ can therefore be summarized by the local wavelet spectral (LWS) matrix \citep{ombao2014}. For each scale $j$ and rescaled time $u\in(0,1)$, the LWS matrix is defined as
\begin{align}
\mathbf{S}_j(u)=\mathbf{V}_j(u)\mathbf{V}_j^\top(u).
\end{align}
The matrix $\mathbf{S}_j(u)$ is symmetric and positive semi-definite. Its $(p,q)$ entry, denoted by $S_j^{(p,q)}(u)$, represents the scale-specific cross-spectrum between channels $p$ and $q$ at rescaled time $u$, thereby characterizing the time-varying dependence structure of the process.

As a methodological precursor to Partial WaveCanCoh, we briefly recall the WaveCanCoh framework of \cite{wu2025wavelet}, which provides a localized measure of (unadjusted) association between two multivariate nonstationary processes. Let
$\mathbf{X}_t=\left(X_t^{(1)}, \ldots, X_t^{(P)}\right)^{\top}$
and
$\mathbf{Y}_t=\left(Y_t^{(1)}, \ldots, Y_t^{(Q)}\right)^{\top}$,
and define the joint multivariate locally stationary process $\{\mathbf{Z}_t\}_{t=1}^T$ with
$\mathbf{Z}_t=\left(\mathbf{X}_t^{\top}, \mathbf{Y}_t^{\top}\right)^{\top}$.
Then, at scale $j$ and rescaled time $u$, the corresponding LWS matrix of $\{\mathbf{Z}_t\}$ admits the block representation
\begin{align} 
\mathbf{S}_{j;\mathbf{ZZ}}(u)
=
\left[\begin{array}{cc}
\mathbf{S}_{j;\mathbf{XX}}(u) & \mathbf{S}_{j;\mathbf{XY}}(u) \\
\mathbf{S}_{j;\mathbf{YX}}(u) & \mathbf{S}_{j;\mathbf{YY}}(u)
\end{array}\right],
\end{align}
where $\mathbf{S}_{j;\mathbf{XX}}(u)$ and $\mathbf{S}_{j;\mathbf{YY}}(u)$ are the auto-LWS matrices of $\{\mathbf{X}_t\}$ and $\{\mathbf{Y}_t\}$, respectively, while $\mathbf{S}_{j;\mathbf{XY}}(u)$ and $\mathbf{S}_{j;\mathbf{YX}}(u)$ describe their localized cross-spectral dependence. Based on this block LWS structure, WaveCanCoh is defined \citep{wu2025wavelet} as the strongest linear association between local linear combinations of $\{\mathbf{X}_t\}$ and $\{\mathbf{Y}_t\}$ at a given scale and rescaled time. Namely, WaveCanCoh between $\{\mathbf{X}_t\}$ and $\{\mathbf{Y}_t\}$ at scale $j$ and rescaled time $u$ is %defined as
\begin{align} \label{equ:WaveCanCoh}
\boldsymbol{\rho}_{j;\mathbf{XY}}(u)
=
\max_{\mathbf{a}_j(u),\,\mathbf{b}_j(u)}
\left\{
\mathbf{a}_j^{\top}(u)\mathbf{S}_{j;\mathbf{XY}}(u)\mathbf{b}_j(u)
\right\}^2,
\end{align}
subject to
\[
\mathbf{a}_j^{\top}(u)\mathbf{S}_{j;\mathbf{XX}}(u)\mathbf{a}_j(u)=1,
\qquad
\mathbf{b}_j^{\top}(u)\mathbf{S}_{j;\mathbf{YY}}(u)\mathbf{b}_j(u)=1,
\]
where $\mathbf{a}_j(u)\in\mathbb{R}^{P}$ and $\mathbf{b}_j(u)\in\mathbb{R}^{Q}$ are the localized canonical weight vectors.

\section{Partial wavelet canonical coherence} \label{sec:partial_wavecancoh}
Let us now move to the setting of interest, where the available data consist of multivariate nonstationary time series and we aim to quantify confounder-adjusted set-to-set association. In this section, we introduce the proposed partial wavelet canonical coherence (Partial WaveCanCoh) framework. To facilitate intuition, we first provide a projection interpretation of the proposed construction, clarifying the logic of partialization in the wavelet domain. We then formally define Partial WaveCanCoh, discuss its key properties, and develop its (consistent) estimation procedure.

\subsection{Motivating the Partial WaveCanCoh construction}
\label{sec:partial_regression}

The goal of this section is to provide a complementary regression-based interpretation of our proposal, intended to clarify the meaning of partialization at a given scale and rescaled time. The formal definition and spectral-domain formulation of Partial WaveCanCoh are then given in Section~\ref{sec:proposal}.
% in terms of scale-specific subprocesses. This interpretation is
To this end, we introduce the notion of a scale-specific subprocess by adapting the subprocess construction proposed in \cite{wu2023multiscale} to a MvLSW process in equation~\eqref{equ: MvLSW}, and define the subprocess of $\{ \mathbf{X}_{t} \}_t$ as the scale-$j$ component $P$-variate locally stationary wavelet process $\{ \mathbf{X}_{j,t} \}_{j,t}$ with
$\mathbf{X}_{j,t}=(X_{j,t}^{(1)},\ldots,X_{j,t}^{(P)})^\top$ represented as
\begin{align}\label{eq:subproc}
    \mathbf{X}_{j,t}= \sum_{k \in \mathbb{Z}} \mathbf{V}_j(k/T)\psi_{j,k}(t)\boldsymbol{\xi}_{j,k}.
\end{align}
%following the MvLSW notation in equation~\eqref{equ: MvLSW}.
%where $\mathbf{V}_j(k/T)$, $\{\psi_{j,k}\}$, and $\{\boldsymbol{\xi}_{j,k}\}$ satisfy Definition~\ref{def:partialwcc} and its associated assumptions. 
Here, $\{X_{j,t}^{(p)}\}_t$ represents the scale-$j$ contribution to the original channel-$p$ process $\{X_{t}^{(p)}\}_t$, and equations~\eqref{equ: MvLSW} and~\eqref{eq:subproc} yield $X_{t}^{(p)}=\sum_{j=1}^{\infty} X_{j,t}^{(p)}$. Similar constructions applied to $\{\mathbf{Y}_t\}_t$ and $\{\mathbf{Z}_t\}_t$ give rise to their corresponding subprocesses, $\{\mathbf{Y}_{j,t}\}_{j,t}$ and $\{\mathbf{Z}_{j,t}\}_{j,t}$.

With this notation in place, the analogy with classical partial canonical correlation becomes transparent by recalling that in the classical setting, one removes from $\mathbf{X}$ and $\mathbf{Y}$ the linear component explained by $\mathbf{Z}$ and then studies the canonical correlation between the resulting residuals. Hence at a {\em fixed scale} $j$, a natural analogue is to interpret the subprocesses $\{\mathbf{X}_{j,t}\}_t$ and $\{\mathbf{Y}_{j,t}\}_t$ as admitting local linear projections onto $\{\mathbf{Z}_{j,t}\}_t$. 

Specifically, for a fixed scale $j$ and rescaled time $u=t/T$, one may represent
\begin{align}\label{eq:regrlsw}
 \mathbf{X}_{j,t} \approx \mathbf{B}^{(X)}_j(u)\mathbf{Z}_{j,t}+\boldsymbol{\varepsilon}^{(X)}_{j,t},
 \mbox{ and }
\mathbf{Y}_{j,t} \approx \mathbf{B}^{(Y)}_j(u)\mathbf{Z}_{j,t}+\boldsymbol{\varepsilon}^{(Y)}_{j,t},   
\end{align}
where $\approx$ denotes equality up to $\mathcal{O}(2^{-j}T^{-1})$, $\mathbf{B}^{(X)}_j(u)$ and $\mathbf{B}^{(Y)}_j(u)$ are scale- and time-dependent coefficient matrices, and $\{\boldsymbol{\varepsilon}^{(X)}_{j,t}\}_t$ and $\{\boldsymbol{\varepsilon}^{(Y)}_{j,t}\}_t$ denote the corresponding residual subprocesses after removing the linear contribution of $\{\mathbf{Z}_{j,t}\}_t$. The source of the order term tracks to \cite{wu2023multiscale}, who have shown that the connection between the subprocess covariance and spectral matrices is given by $$\operatorname{E}(\mathbf{X}_{j,t} \mathbf{Y}^T_{j,t}) \approx \mathbf{S}_{j; \mathbf{XY}}(u), \mbox{ for any subprocesses } \{\mathbf{X}_{j,t}\}_t, \, \{\mathbf{Y}_{j,t}\}_t.$$ Under this interpretation, $\mathbf{B}^{(X)}_j(u)$ and $\mathbf{B}^{(Y)}_j(u)$ play the role of local population least-squares coefficients and at the level of second-order structure, their analogues are
\begin{align*}
 \mathbf{B}^{(X)}_j(u)=\mathbf{S}_{j;\mathbf{XZ}}(u)\mathbf{S}_{j;\mathbf{ZZ}}^{-1}(u),
\qquad
\mathbf{B}^{(Y)}_j(u)=\mathbf{S}_{j;\mathbf{YZ}}(u)\mathbf{S}_{j;\mathbf{ZZ}}^{-1}(u),   
\end{align*}
provided that $\mathbf{S}_{j;\mathbf{ZZ}}(u)$ is nonsingular. 
Akin to Section~\ref{sec:partial_cca}, substituting the projection coefficients into the residualization step~\eqref{eq:regrlsw} leads to the residual second-order quantities of the type $\operatorname{E}(\boldsymbol{\varepsilon}^{(\cdotp)}_{j,t} \boldsymbol{\varepsilon}^{(\cdotp), \top}_{j,t})$ for process combinations $(\mathbf{X},\mathbf{X}), \,  (\mathbf{Y},\mathbf{Y}), \, (\mathbf{X},\mathbf{Y})$. Specifically, % These can be expressed as 
\begin{align*}
\mathbf{S}_{j;\mathbf{XX \cdot Z}}(u)
&=
\mathbf{S}_{j;\mathbf{XX}}(u)
-
\mathbf{S}_{j;\mathbf{XZ}}(u)\mathbf{S}_{j;\mathbf{ZZ}}^{-1}(u)\mathbf{S}_{j;\mathbf{ZX}}(u), \text{ and }   \\ 
   \mathbf{S}_{j;\mathbf{YY \cdot Z}}(u)
&=
\mathbf{S}_{j;\mathbf{YY}}(u)
-
\mathbf{S}_{j;\mathbf{YZ}}(u)\mathbf{S}_{j;\mathbf{ZZ}}^{-1}(u)\mathbf{S}_{j;\mathbf{ZY}}(u), 
\end{align*}
while the residual cross-dependence between $\{\mathbf{X}_{j,t}\}_t$ and $\{\mathbf{Y}_{j,t}\}_t$ is represented by
\begin{align*}
    \mathbf{S}_{j;\mathbf{XY \cdot Z}}(u)
=
\mathbf{S}_{j;\mathbf{XY}}(u)
-
\mathbf{S}_{j;\mathbf{XZ}}(u)\mathbf{S}_{j;\mathbf{ZZ}}^{-1}(u)\mathbf{S}_{j;\mathbf{ZY}}(u).
\end{align*}
We refer to these as the {\em partial local wavelet spectral matrices}. Consequently, the proposed partialization may be understood as removing, within each scale separately, the component of the local second-order dependence structure of $\{\mathbf{X}_{j,t}\}$ and $\{\mathbf{Y}_{j,t}\}$ that is linearly explained by $\{\mathbf{Z}_{j,t}\}$. Partial WaveCanCoh then measures the canonical coherence between the resulting residual subprocesses in this scale-specific second-order sense. %introduced in the main text. 

We stress that the above display is not intended as a strict generative regression model for the subprocesses. Rather, it provides a projection-based interpretation of the Schur-complement construction underlying
$\mathbf{S}_{j;\mathbf{XX \cdot Z}}(u)$,
$\mathbf{S}_{j;\mathbf{YY \cdot Z}}(u)$, and
$\mathbf{S}_{j;\mathbf{XY \cdot Z}}(u)$ shown in next part. In particular, the proposed procedure performs partialization within each scale $j$ separately: the contribution of the scale-$j$ subprocess of $\{\mathbf{Z}_t\}$ is removed from the scale-$j$ subprocesses of $\{\mathbf{X}_t\}$ and $\{\mathbf{Y}_t\}$. In this way, Partial WaveCanCoh provides a direct scale-specific analogue of the residualization step in classical partial canonical correlation analysis.

\subsection{Proposed partial wavelet canonical coherence (Partial WaveCanCoh)} \label{sec:proposal}

We now introduce the proposed partial wavelet canonical coherence (Partial WaveCanCoh) measure, which extends classical partial canonical correlation analysis to multivariate nonstationary time series by making use of the MvLSW framework reviewed in Section~\ref{sec:wavecancoh}. 

Suppose we have two sets of multivariate nonstationary time series $\{\mathbf{X}_t\}$ and $\{\mathbf{Y}_t\}$ of potentially different dimensions,  namely $\mathbf{X}_t=\left(X_t^{(1)}, \ldots, X_{t}^{(P)}\right)^{\top}$ and $\mathbf{Y}_t=\left(Y_t^{(1)},  \ldots, Y_{t}^{(Q)}\right)^{\top}$, as well as a third confounding process, $\{\mathbf{Z}_t\}$ with $\mathbf{Z}_t=\left(Z_t^{(1)}, \ldots, Z_{t}^{(K)}\right)^{\top}$. 
%for which we adopt the locally stationary wavelet processes as in equation \eqref{equ: MvLSW}, %locally stationary wavelet 

We define the joint $(P+Q+K)$-dimensional process $\{\mathbf{W}_t\}$ with 
$\mathbf{W}_t
=
\bigl(
\mathbf{X}_t^\top,\,
\mathbf{Y}_t^\top,\,
\mathbf{Z}_t^\top
\bigr)^\top,
$ for which we adopt the multivariate locally stationary wavelet process structure as in equation \eqref{equ: MvLSW}, 
whose (symmetrical) spectrum at scale $j$ and rescaled time $u$ is % $\mathbf{S}_j(u)$ %obtained as
\begin{align} \label{equ:cross-spec}
    \mathbf{S}_{j}(u) =  \left[\begin{array}{lll}
\mathbf{S}_{j; \mathbf{X X}}(u) & \mathbf{S}_{j; \mathbf{X Y}}(u)  & \mathbf{S}_{j; \mathbf{X Z}}(u)\\
\mathbf{S}_{j; \mathbf{Y X}}(u) & \mathbf{S}_{j; \mathbf{Y Y}}(u) & \mathbf{S}_{j; \mathbf{YZ}}(u) \\
\mathbf{S}_{j; \mathbf{Z X}}(u) & \mathbf{S}_{j; \mathbf{Z Y}}(u) & \mathbf{S}_{j; \mathbf{ZZ}}(u) \\
\end{array}\right]_{(P + Q + K) \times (P + Q + K)}
\end{align}
% In equation \ref{equ:cross-spec}, the main diagonal blocks $\mathbf{S}_{j; \mathbf{X X}}(u)$ $(P \times P)$, $\mathbf{S}_{j; \mathbf{YY}}(u)$ $(Q \times Q)$ and $\mathbf{S}_{j; \mathbf{ZZ}}(u)$ $(K \times K)$ denote the auto-LWS matrices of  the $\{\mathbf{X}_t\}$, $\{\mathbf{Y}_t\}$ and $\{\mathbf{Z}_t\}$   processes, respectively, while the off-diagonal blocks denote their cross-LWS matrices. 
where $\mathbf{S}_{j;\mathbf{XX}}(u)$, $\mathbf{S}_{j;\mathbf{YY}}(u)$, and $\mathbf{S}_{j;\mathbf{ZZ}}(u)$ are the auto-local wavelet spectral matrices of $\{\mathbf{X}_t\}$, $\{\mathbf{Y}_t\}$, and $\{\mathbf{Z}_t\}$, respectively, and the off-diagonal blocks denote the corresponding cross-local wavelet spectral matrices. The LWS of the meta-process $\{\mathbf{W}_t \}$ provides a time-localized contribution of individual channels and cross-channels to the process covariance. 

\begin{remark} For now, we assume that $\mathbf{S}_{j;\mathbf{ZZ}}(u)$ is nonsingular for all scale-time pairs $(j,u)$ of interest and that obtaining its inverse is feasible. Section~\ref{sec:high_dim} deals with infeasible setups.\end{remark}
In analogy with classical partial canonical correlation analysis, we remove the linear contribution of $\{\mathbf{Z}_t\}$ from the local second-order dependence structure of $\{\mathbf{X}_t\}$ and $\{\mathbf{Y}_t\}$ by means of the partial local wavelet spectral matrices, as follows.

%\noindent \textbf{Definition 1 (Localized Scale-specific Partial Wavelet Canonical Coherence).}
\begin{definition}\label{def:partialwcc}
Let $\{\bigl(
\mathbf{X}_t^\top,\,
\mathbf{Y}_t^\top,\,
\mathbf{Z}_t^\top
\bigr)^\top\}_{t=1}^T$ be a jointly multivariate locally stationary time series, where $\mathbf{X}_t=\left(X_t^{(1)}, \ldots, X_{t}^{(P)}\right)^{\top}$, $\mathbf{Y}_t=\left(Y_t^{(1)}, \ldots, Y_{t}^{(Q)}\right)^{\top}$, and $\mathbf{Z}_t=\left(Z_t^{(1)}, \ldots, Z_{t}^{(K)}\right)^{\top}$. The localized scale-specific partial wavelet canonical coherence (Partial WaveCanCoh) between processes $\{\mathbf{X}_t\}$ and $\{\mathbf{Y}_t\}$ after removing the linear contribution of $\{\mathbf{Z}_t\}$, at scale $j$ and rescaled time $u$, is defined as
\begin{align} \label{equ: partial WaveCanCoh}
    \boldsymbol{\rho}_{j;\mathbf{XY \cdot Z}}(u)
    =
    \max_{\mathbf{a}_j(u), \mathbf{b}_j(u)}
    \left\{
    \mathbf{a}_j^{\top}(u)\mathbf{S}_{j;\mathbf{XY \cdot Z}}(u)\mathbf{b}_j(u)
    \right\}^{2},
\end{align}
subject to the constraints
\[
\mathbf{a}_j^{\top}(u)\mathbf{S}_{j;\mathbf{XX \cdot Z}}(u)\mathbf{a}_j(u)=1
\qquad \text{and} \qquad
\mathbf{b}_j^{\top}(u)\mathbf{S}_{j;\mathbf{YY \cdot Z}}(u)\mathbf{b}_j(u)=1.
\]
In the above, $\mathbf{a}_j^{\top}(u)=\left(a_j^{(p)}(u)\right)_{p=1}^{P}$ is a $1\times P$ vector and $\mathbf{b}_j^{\top}(u)=\left(b_j^{(q)}(u)\right)_{q=1}^{Q}$ is a $1\times Q$ vector, representing the localized partial canonical vectors associated with $\{\mathbf{X}_t\}$ and $\{\mathbf{Y}_t\}$, respectively. Furthermore, the partial local spectral matrices are defined as 
\begin{align}
\mathbf{S}_{j;\mathbf{XX}\cdot\mathbf{Z}}(u)
&:=
\mathbf{S}_{j;\mathbf{XX}}(u)
-
\mathbf{S}_{j;\mathbf{XZ}}(u)
\mathbf{S}_{j;\mathbf{ZZ}}^{-1}(u)
\mathbf{S}_{j;\mathbf{ZX}}(u),
\label{eq:partial_lws_xx}
\\
\mathbf{S}_{j;\mathbf{YY}\cdot\mathbf{Z}}(u)
&:=
\mathbf{S}_{j;\mathbf{YY}}(u)
-
\mathbf{S}_{j;\mathbf{YZ}}(u)
\mathbf{S}_{j;\mathbf{ZZ}}^{-1}(u)
\mathbf{S}_{j;\mathbf{ZY}}(u),
\label{eq:partial_lws_yy}
\\
\mathbf{S}_{j;\mathbf{XY}\cdot\mathbf{Z}}(u)
&:=
\mathbf{S}_{j;\mathbf{XY}}(u)
-
\mathbf{S}_{j;\mathbf{XZ}}(u)
\mathbf{S}_{j;\mathbf{ZZ}}^{-1}(u)
\mathbf{S}_{j;\mathbf{ZY}}(u),
\label{eq:partial_lws_xy}
\end{align}
and represent the residual scale-dependent second-order structure between $\{\mathbf{X}_t\}$ and $\{\mathbf{Y}_t\}$ after accounting for the contribution of $\{\mathbf{Z}_t\}$ at scale $j$ and rescaled time $u$. (Note the relationship is symmetrical as $
\mathbf{S}_{j;\mathbf{YX}\cdot\mathbf{Z}}(u)
:=\mathbf{S}_{j;\mathbf{YX}}(u)
-
\mathbf{S}_{j;\mathbf{YZ}}(u)
\mathbf{S}_{j;\mathbf{ZZ}}^{-1}(u)
\mathbf{S}_{j;\mathbf{ZX}}(u)= \mathbf{S}_{j;\mathbf{XY}\cdot\mathbf{Z}}^{\top}(u)
\label{eq:partial_lws_yx}$.)
\end{definition}

%\noindent \textit{\textbf{Remark 1}. {}
\begin{remark}The function $\boldsymbol{\rho}_{j;\mathbf{XY \cdot Z}}(\cdotp)$ provides a time-varying measure of the direct linear association between $\{\mathbf{X}_t\}$ and $\{\mathbf{Y}_t\}$ at scale $j$, after accounting for the linear effect of $\{\mathbf{Z}_t\}$. It takes values in $[0,1]$, where values close to 1 indicate strong residual linear dependence and values close to 0 indicate weak or negligible residual linear dependence. Moreover, the entries ${a}_j^{(p)}(\cdotp)$ and ${b}_j^{(q)}(\cdotp)$ describe the time-localized contributions of the $p$th channel in $\{\mathbf{X}_t\}$ and the $q$th channel in $\{\mathbf{Y}_t\}$ to the resulting partial wavelet canonical coherence.\end{remark}

An equivalent characterization of the optimization problem in Definition~\ref{def:partialwcc} can be obtained through an eigenvalue--eigenvector formulation. Specifically, the localized partial canonical vectors $\mathbf{a}_j(\cdotp)$ and $\mathbf{b}_j(\cdotp)$ are obtained by the eigendecompositions associated with the matrices
\begin{align}
    \mathbf{N}_{j;\mathbf{a}\mid \mathbf{Z}}(u)
    &=
    \mathbf{S}_{j;\mathbf{XX \cdot Z}}^{-1}(u)
    \mathbf{S}_{j;\mathbf{XY \cdot Z}}(u)
    \mathbf{S}_{j;\mathbf{YY \cdot Z}}^{-1}(u)
    \mathbf{S}_{j;\mathbf{YX \cdot Z}}(u),
    \label{equ:partial_a} 
    \\
    \mathbf{N}_{j;\mathbf{b}\mid \mathbf{Z}}(u)
    &=
    \mathbf{S}_{j;\mathbf{YY \cdot Z}}^{-1}(u)
    \mathbf{S}_{j;\mathbf{YX \cdot Z}}(u)
    \mathbf{S}_{j;\mathbf{XX \cdot Z}}^{-1}(u)
    \mathbf{S}_{j;\mathbf{XY \cdot Z}}(u).
    \label{equ:partial_b}
\end{align}

\begin{proposition}
    Let $\Lambda_{j;\mathbf{a}\mid \mathbf{Z}}^{(k)}(u)$ denote the $k$th largest eigenvalue of $\mathbf{N}_{j;\mathbf{a}\mid \mathbf{Z}}(u)$ in \eqref{equ:partial_a}, and let $\Lambda_{j;\mathbf{b}\mid \mathbf{Z}}^{(l)}(u)$ denote the $l$th largest eigenvalue of $\mathbf{N}_{j;\mathbf{b}\mid \mathbf{Z}}(u)$ in \eqref{equ:partial_b}, for $k,l=1,\ldots,\min(P,Q)$. The matrices $\mathbf{N}_{j;\mathbf{a}\mid \mathbf{Z}}(u)$ and $\mathbf{N}_{j;\mathbf{b}\mid \mathbf{Z}}(u)$ have the same non-zero eigenvalues, and denoting the largest eigenvalue as $
\Lambda_{j;\mathbf{XY \cdot Z}}^{(1)}(u)
:=
\Lambda_{j;\mathbf{a}\mid \mathbf{Z}}^{(1)}(u)
=
\Lambda_{j;\mathbf{b}\mid \mathbf{Z}}^{(1)}(u),
$
then the localized scale-specific partial wavelet canonical coherence in \eqref{equ: partial WaveCanCoh} admits the equivalent representation
\begin{align}
    \boldsymbol{\rho}_{j;\mathbf{XY \cdot Z}}(u)
    =
    \Lambda_{j;\mathbf{XY \cdot Z}}^{(1)}(u).
    \label{equ:partial_rho_lambda}
\end{align}
Moreover, the eigenvectors of $\mathbf{N}_{j;\mathbf{a}\mid \mathbf{Z}}(u)$ and $\mathbf{N}_{j;\mathbf{b}\mid \mathbf{Z}}(u)$ corresponding to the largest eigenvalue $\Lambda_{j;\mathbf{XY \cdot Z}}^{(1)}(u)$ determine the localized partial canonical vectors $\mathbf{a}_j(u)$, $\mathbf{b}_j(u)$ for $\{\mathbf{X}_t\}$ and $\{\mathbf{Y}_t\}$, respectively, at given scale $j$ and rescaled time $u$.
\end{proposition}
% since $\mathbf{N}_{j;\mathbf{b}\mid \mathbf{Z}}(u)=\mathbf{N}_{j;\mathbf{a}\mid \mathbf{Z}}^{\top}(u)$

 \textit{Proof: See Appendix \ref{app:proof}}.

\subsection{Estimation procedure}
\label{subsec:estimation}

We next consider the estimation of the proposed Partial WaveCanCoh measure in equation~\eqref{equ: partial WaveCanCoh}, or equivalently~\eqref{equ:partial_rho_lambda}. 
\begin{comment}
Let
$\{\mathbf{X}_t\}$, $\{\mathbf{Y}_t\}$, and $\{\mathbf{Z}_t\}$ be nonstationary time series that can jointly modeled as an MvLSW concatenated process, defined as
\[
\mathbf{W}_t=\bigl(\mathbf{X}_t^\top,\mathbf{Y}_t^\top,\mathbf{Z}_t^\top\bigr)^\top.
\]
\end{comment}
As is evident from the preceding development, the central ingredient is the joint local wavelet spectral matrix of $\{\mathbf{W}_t=\bigl(\mathbf{X}_t^\top,\mathbf{Y}_t^\top,\mathbf{Z}_t^\top\bigr)^\top\}$. Indeed, once an estimator of the joint LWS matrix defined in equation \eqref{equ:cross-spec} is available, all partial local wavelet spectral matrices required for Partial WaveCanCoh follow immediately through the Schur-complement constructions introduced in equations~\eqref{eq:partial_lws_xx},~\eqref{eq:partial_lws_yy} and~\eqref{eq:partial_lws_xy}. The estimation of Partial WaveCanCoh therefore reduces to the estimation of the joint second-order wavelet structure of the meta-process $\{\mathbf{W}_t\}$.

In this paper, the local wavelet spectral matrix is estimated following the methodology of \cite{nason2000wavelet, ombao2014}, which has become standard in the locally stationary wavelet literature and has been used extensively for a variety of data contexts. We give a brief summary of the main steps here and refer the reader to these references for a more detailed treatment.

Let $\widehat{\mathbf{S}}_{j,k}$ denote the estimator of the joint LWS matrix at scale $j=1, \ldots, J=\log_2(T)$ and rescaled time $k/T$, defined by the corrected smoothed periodogram
\begin{align}
    \widehat{\mathbf{S}}_{j,k}
    =
    \sum_{l=1}^{J} (\mathbf{A}^{-1})_{jl}\widetilde{\mathbf{I}}_{l,k},
    \label{eq:joint_lws_est}
\end{align}
where $A_{jl}=\langle \Psi_j,\Psi_l\rangle$ is the $(j,l)$ entry of the inner product matrix $\mathbf{A}$ of the discrete autocorrelation wavelets \citep{nason2000wavelet}, and the matrix 
\begin{align}
    \widetilde{\mathbf{I}}_{l,k}
    =
    \frac{1}{2M+1}\sum_{m=-M}^{M}\mathbf{I}_{l,k+m}
    \label{eq:smoothed_periodogram}
\end{align}
denotes the smoothed periodogram at scale $l$. Here, $M$ is the half-width of the rectangular smoothing kernel, typically taken to be $M=\lfloor \sqrt{T} \rfloor$, and $k=\lfloor uT\rfloor$ indexes the time location corresponding to rescaled time $u$. The scale-$l$ and time $k$ raw periodogram matrix $\mathbf{I}_{l,k}$ is constructed from the empirical wavelet coefficient vector $\mathbf{d}_{l,k}$ of the joint process as
\begin{align}
    \mathbf{I}_{l,k}
    =
    \mathbf{d}_{l,k}\mathbf{d}_{l,k}^{\top}, \mbox{ where }
    \mathbf{d}_{l,k}
    =
    \sum_{t=1}^{T}\mathbf{W}_t\psi_{l,k}(t).
    \label{eq:raw_periodogram_joint}
\end{align}
Under the usual asymptotic conditions $T,M\to\infty$ and $M/T\to 0$, $\widehat{\mathbf{S}}_{j,k}$ is a consistent estimator of the joint LWS matrix $\mathbf{S}_j(u)$ at each $(j, k=\lfloor uT\rfloor)$ pair \citep{ombao2014}.

From the $\widehat{\mathbf{S}}_{j,k}$ partitioning mirroring~\eqref{equ:cross-spec}, we extract the estimated block matrices
$\widehat{\mathbf{S}}_{j;\mathbf{XX}}(u)$,
$\widehat{\mathbf{S}}_{j;\mathbf{YY}}(u)$,
$\widehat{\mathbf{S}}_{j;\mathbf{ZZ}}(u)$,
$\widehat{\mathbf{S}}_{j;\mathbf{XY}}(u)$,
$\widehat{\mathbf{S}}_{j;\mathbf{XZ}}(u)$, and
$\widehat{\mathbf{S}}_{j;\mathbf{YZ}}(u)$ (and their transposed versions), and obtain the estimated partial local wavelet spectral matrices as
\begin{align}
    \widehat{\mathbf{S}}_{j;\mathbf{XX \cdot Z}}(u)
    &=
    \widehat{\mathbf{S}}_{j;\mathbf{XX}}(u)
    -
    \widehat{\mathbf{S}}_{j;\mathbf{XZ}}(u)
    \widehat{\mathbf{S}}_{j;\mathbf{ZZ}}^{-1}(u)
    \widehat{\mathbf{S}}_{j;\mathbf{ZX}}(u), \label{eq:localpartialspec1}
    \\
    \widehat{\mathbf{S}}_{j;\mathbf{YY \cdot Z}}(u)
    &=
    \widehat{\mathbf{S}}_{j;\mathbf{YY}}(u)
    -
    \widehat{\mathbf{S}}_{j;\mathbf{YZ}}(u)
    \widehat{\mathbf{S}}_{j;\mathbf{ZZ}}^{-1}(u)
    \widehat{\mathbf{S}}_{j;\mathbf{ZY}}(u),\label{eq:localpartialspec2}
    \\
    \widehat{\mathbf{S}}_{j;\mathbf{XY \cdot Z}}(u)
    &=
    \widehat{\mathbf{S}}_{j;\mathbf{XY}}(u)
    -
    \widehat{\mathbf{S}}_{j;\mathbf{XZ}}(u)
    \widehat{\mathbf{S}}_{j;\mathbf{ZZ}}^{-1}(u)
    \widehat{\mathbf{S}}_{j;\mathbf{ZY}}(u).\label{eq:localpartialspec3}
    %\\
    %\widehat{\mathbf{S}}_{j;\mathbf{YX \cdot Z}}(u)
    %&=
    %\widehat{\mathbf{S}}_{j;\mathbf{YX}}(u)
    %-
    %\widehat{\mathbf{S}}_{j;\mathbf{YZ}}(u)
    %\widehat{\mathbf{S}}_{j;\mathbf{ZZ}}^{-1}(u)
    %\widehat{\mathbf{S}}_{j;\mathbf{ZX}}(u).
\end{align}

\begin{proposition}
Let $\{\mathbf{W}_t=\bigl(\mathbf{X}_t^\top,\mathbf{Y}_t^\top,\mathbf{Z}_t^\top\bigr)^\top\}$ be a multivariate locally stationary process with uniform absolutely summable autocovariance and whose true unknown spectral structure $\{\mathbf{S}_j(\cdotp)\}_j$ may be estimated using equation~\eqref{eq:joint_lws_est}, and its partial
local spectral matrices $\{\mathbf{S}_{j;\mathbf{AB\cdot Z}}(\cdotp)\}_j$ may be estimated as in equations~\eqref{eq:localpartialspec1},~\eqref{eq:localpartialspec2} and~\eqref{eq:localpartialspec3}, for any $\mathbf{A},\mathbf{B} \in \{ \mathbf{X}, \mathbf{Y}\}$. 

Then, we propose the estimator of the localized scale-specific partial wavelet canonical coherence that measures the direct association between the processes $\{\mathbf{X}_t\}$ and $\{\mathbf{Y}_t\}$ having accounted for the linear effect of the confounding process $\{\mathbf{Z}_t\}$ to be% at each scale $j$ and rescaled time $u$,
\begin{align}
    \widehat{\boldsymbol{\rho}}_{j;\mathbf{XY \cdot Z}}(u)
    =
    \widehat{\Lambda}_{j;\mathbf{XY \cdot Z}}^{(1)}(u),
    \label{eq:partial_wavecancoh_est}
\end{align}
where $\widehat{\Lambda}_{j;\mathbf{XY \cdot Z}}^{(1)}(u)$ is the largest eigenvalue of
\begin{align}
    \widehat{\mathbf{N}}_{j;\mathbf{a}\mid \mathbf{Z}}(u)
    &=
    \widehat{\mathbf{S}}_{j;\mathbf{XX \cdot Z}}^{-1}(u)
    \widehat{\mathbf{S}}_{j;\mathbf{XY \cdot Z}}(u)
    \widehat{\mathbf{S}}_{j;\mathbf{YY \cdot Z}}^{-1}(u)
    \widehat{\mathbf{S}}_{j;\mathbf{YX \cdot Z}}(u),
    \label{eq:partial_M_hat_a}
    \\
    \widehat{\mathbf{N}}_{j;\mathbf{b}\mid \mathbf{Z}}(u)
    &=
    \widehat{\mathbf{S}}_{j;\mathbf{YY \cdot Z}}^{-1}(u)
    \widehat{\mathbf{S}}_{j;\mathbf{YX \cdot Z}}(u)
    \widehat{\mathbf{S}}_{j;\mathbf{XX \cdot Z}}^{-1}(u)
    \widehat{\mathbf{S}}_{j;\mathbf{XY \cdot Z}}(u).
    \label{eq:partial_M_hat_b}
\end{align}
The corresponding estimators of the localized partial canonical vectors, $\widehat{\mathbf{a}}_j(u)$ and $\widehat{\mathbf{b}}_j(u)$, are the eigenvectors of $\widehat{\mathbf{N}}_{j;\mathbf{a}\cdot \mathbf{Z}}(u)$ and $\widehat{\mathbf{N}}_{j;\mathbf{b}\cdot \mathbf{Z}}(u)$ associated with $\widehat{\Lambda}_{j;\mathbf{XY \cdot Z}}^{(1)}(u)$.

As the number of time points $T \rightarrow \infty$ and the smoothing
parameter $M \rightarrow \infty$ with $M/T \rightarrow 0$, the proposed estimators are consistent with the true quantities they aim to approximate. 
\end{proposition}
 \textit{Proof: See Appendix \ref{app:proof}}.

In a nutshell, estimation of Partial WaveCanCoh proceeds by first estimating the joint second-order wavelet structure of the concatenated meta-process $\{\mathbf{W}_t=\bigl(\mathbf{X}_t^\top,\mathbf{Y}_t^\top,\mathbf{Z}_t^\top\bigr)^\top\}$, and then applying the partialization step in equations~\eqref{eq:partial_lws_xx}--\eqref{eq:partial_lws_xy} taken at the level of the estimated local wavelet spectral matrices. Under standard regularity conditions, including increasing sample size and appropriate temporal smoothing, the resulting estimators are consistent for their population partial canonical coherence and corresponding vectors counterparts, a property that we will also numerically illustrate in Section~\ref{sec:sim}. %\noindent\textit{Proof: See Appendix~\ref{app:estimation}.}

\section{Extension to high-dimensional confounders}
\label{sec:high_dim}

The developed framework assumes that the confounding process $\{\mathbf{Z}_t\}$ can be incorporated directly into the partialization step through the matrix $\mathbf{S}_{j;\mathbf{ZZ}}(u)$. In many applications, however, $\{\mathbf{Z}_t\}$ may be high-dimensional. This occurs naturally in neural, financial, and environmental settings, where a large number of auxiliary variables may simultaneously influence both processes $\{\mathbf{X}_t\}$ and $\{\mathbf{Y}_t\}$. In such cases, direct implementation of Partial WaveCanCoh may become unstable, since estimation of $\mathbf{S}_{j;\mathbf{ZZ}}(u)$ can be unreliable when its dimension $K$ is large relative to the available sample size, and the inversion of $\mathbf{S}_{j;\mathbf{ZZ}}(u)$ required for partialization may be ill-conditioned or infeasible. As a result, estimation of the partial local wavelet spectral matrices may deteriorate substantially.

To address this issue, we consider a factor-based extension of the proposed framework. Rather than conditioning on the full high-dimensional confounding process $\{\mathbf{Z}_t\}$, we replace it by a lower-dimensional representation that captures its dominant common variation. Specifically, suppose that
\begin{align}
    \mathbf{Z}_t = \mathbf{L}\mathbf{f}_t + \boldsymbol{\eta}_t,
    \label{eq:factor_model_Z}
\end{align}
where $\mathbf{f}_t \in \mathbb{R}^r$ is an $r$-dimensional latent factor process with $r \ll K$, $\mathbf{L}$ is a $K\times r$ loading matrix, and $\boldsymbol{\eta}_t$ is an idiosyncratic component. The factor process $\{\mathbf{f}_t\}$ is intended to summarize the low-dimensional component that carries the dominant confounding effect. % of $\{\mathbf{Z}_t\}$

In this paper, we estimate the factor process using principal component analysis (PCA) applied to the observed confounding series $\{\mathbf{Z}_t\}$; see, for example, \cite{jolliffe2002principal}.
% \textcolor{red}{why not LSW-PCA? (JCGS)}
We then replace the original confounding process $\{\mathbf{Z}_t\}$ by the resulting estimated factor process $\{\widehat{\mathbf{f}}_t\}$ and work with the reduced concatenated process
\begin{align*}
    \widetilde{\mathbf{W}}_t
    =
    \bigl(\mathbf{X}_t^\top,\mathbf{Y}_t^\top,\widehat{\mathbf{f}}_t^\top\bigr)^\top.
\end{align*}
The localized scale-specific partial wavelet canonical coherence is estimated using the corrected smoothed periodogram of $\{\widetilde{\mathbf{W}}_t\}$ following the strategy proposed in Section \ref{subsec:estimation}, and the subsequent partialization step is carried out with respect to the reduced factor process rather than the original high-dimensional confounders. For example, we use
\begin{align}
    \mathbf{S}_{j;\mathbf{XX}\cdot \widehat{\mathbf{f}}}(u)
    =
    \mathbf{S}_{j;\mathbf{XX}}(u)
    -
    \mathbf{S}_{j;\mathbf{X}\widehat{\mathbf{f}}}(u)
    \mathbf{S}_{j;\widehat{\mathbf{f}}\widehat{\mathbf{f}}}^{-1}(u)
    \mathbf{S}_{j;\widehat{\mathbf{f}}\mathbf{X}}(u),
    \label{eq:partial_factor_xx}
\end{align}
with analogous expressions for
$\mathbf{S}_{j;\mathbf{YY}\cdot \widehat{\mathbf{f}}}(u)$,
$\mathbf{S}_{j;\mathbf{XY}\cdot \widehat{\mathbf{f}}}(u)$, and
$\mathbf{S}_{j;\mathbf{YX}\cdot \widehat{\mathbf{f}}}(u)$. The resulting estimated Partial WavCanCoh measure may be interpreted as the direct scale-specific association between $\{\mathbf{X}_t\}$ and $\{\mathbf{Y}_t\}$ after removing the contribution explained by the dominant low-dimensional structure present in $\{\mathbf{Z}_t\}$.

This extension is appealing for several reasons. First, it preserves the structure of the proposed Partial WaveCanCoh framework: after factor extraction, all subsequent steps remain unchanged. Second, it improves numerical stability by reducing the dimension of the matrix to be inverted in the partialization step. Third, it is well suited to applications in which a large collection of observed confounders is in fact driven by a relatively small number of common sources. In such cases, conditioning on the leading principal components may be both more stable and more interpretable than conditioning on all observed confounders.% individually.

Naturally, the effectiveness of this strategy depends on whether the confounding effect of $\{\mathbf{Z}_t\}$ is primarily driven by a low-rank common component. If relevant confounding information is instead concentrated in highly localized idiosyncratic components, the PCA-based reduction may yield only partial adjustment. Nevertheless, as a practical extension, the proposed PCA-based approach offers an effective and computationally scalable solution for high-dimensional confounding while remaining fully compatible with the wavelet-based nonstationary framework developed in this paper. Other factor-extraction methods for dependent data may also be incorporated in the same manner, but are not pursued here.

Algorithm \ref{alg:alg_partial_hd} summarizes the proposed estimation procedure for the localized partial wavelet canonical coherence. The PCA step is optional and is mainly intended for settings in which the confounder dimension is large or the estimation of $\mathbf{S}_{j;\mathbf{ZZ}}(u)$ is unstable.

\begin{algorithm}[htbp]
\caption{Partial WaveCanCoh estimation with optional PCA reduction}
\label{alg:alg_partial_hd}
\begin{algorithmic}
\State Observe $\mathbf{X}_t=\left(X_t^{(1)}, \ldots, X_{t}^{(P)}\right)^{\top}$,
$\mathbf{Y}_t=\left(Y_t^{(1)}, \ldots, Y_{t}^{(Q)}\right)^{\top}$, and
$\mathbf{Z}_t=\left(Z_t^{(1)}, \ldots, Z_{t}^{(K)}\right)^{\top}$ for $t=1,\ldots,T$.

\State \textbf{1. Optional PCA reduction:}
if $K$ is large, apply principal component analysis to $\{\mathbf{Z}_t\}$ and retain the leading $r$ components to form a low-dimensional factor process $\{\widehat{\mathbf{f}}_t\}$; otherwise retain $\{\mathbf{Z}_t\}$ directly.

\State \textbf{2. Fuse:}
construct the joint process
$\widetilde{\mathbf{W}}_t=\left(\mathbf{X}_t^\top,\mathbf{Y}_t^\top,\tilde{\mathbf{Z}}_t^\top\right)^\top$,
where $\tilde{\mathbf{Z}}_t:=\widehat{\mathbf{f}}_t$ if PCA reduction is used, and $\tilde{\mathbf{Z}}_t:=\mathbf{Z}_t$ otherwise.

\State \textbf{3. Spectral estimation:}
estimate the joint LWS matrix of $\{\widetilde{\mathbf{W}}_t\}$ using \eqref{eq:joint_lws_est} at each scale $j=1, \ldots, J$ and rescaled time $u=t/T$, and extract using the structure in~\eqref{equ:cross-spec} the block matrices %^\ast
$\widehat{\mathbf{S}}_{j;\mathbf{XX}}(u)$,
$\widehat{\mathbf{S}}_{j;\mathbf{YY}}(u)$,
$\widehat{\mathbf{S}}_{j;\mathbf{X}\tilde{\mathbf{Z}}}(u)$,
$\widehat{\mathbf{S}}_{j;\mathbf{Y}\tilde{\mathbf{Z}}}(u)$, and
$\widehat{\mathbf{S}}_{j;\mathbf{\tilde{\mathbf{Z}}\tilde{\mathbf{Z}}}}(u)$.

\State \textbf{4. Partialization:}
compute
\begin{align*}
\widehat{\mathbf{S}}_{j;\mathbf{XX}\cdot \tilde{\mathbf{Z}}}(u)
&=
\widehat{\mathbf{S}}_{j;\mathbf{XX}}(u)
-
\widehat{\mathbf{S}}_{j;\mathbf{X}\tilde{\mathbf{Z}}}(u)
\widehat{\mathbf{S}}_{j;\mathbf{\tilde{\mathbf{Z}}\tilde{\mathbf{Z}}}}^{-1}(u)
\widehat{\mathbf{S}}_{j;\mathbf{\tilde{\mathbf{Z}}X}}(u),\\
\widehat{\mathbf{S}}_{j;\mathbf{YY}\cdot \tilde{\mathbf{Z}}}(u)
&=
\widehat{\mathbf{S}}_{j;\mathbf{YY}}(u)
-
\widehat{\mathbf{S}}_{j;\mathbf{Y\tilde{\mathbf{Z}}}}(u)
\widehat{\mathbf{S}}_{j;\mathbf{\tilde{\mathbf{Z}}\tilde{\mathbf{Z}}}}^{-1}(u)
\widehat{\mathbf{S}}_{j;\mathbf{\tilde{\mathbf{Z}}Y}}(u),\\
\widehat{\mathbf{S}}_{j;\mathbf{XY}\cdot \tilde{\mathbf{Z}}}(u)
&=
\widehat{\mathbf{S}}_{j;\mathbf{XY}}(u)
-
\widehat{\mathbf{S}}_{j;\mathbf{X\tilde{\mathbf{Z}}}}(u)
\widehat{\mathbf{S}}_{j;\mathbf{\tilde{\mathbf{Z}}\tilde{\mathbf{Z}}}}^{-1}(u)
\widehat{\mathbf{S}}_{j;\mathbf{\tilde{\mathbf{Z}}Y}}(u),
\end{align*}
together with $\widehat{\mathbf{S}}_{j;\mathbf{YX}\cdot \tilde{\mathbf{Z}}}(u)$.

\State \textbf{5. Eigendecomposition:}
form
$\widehat{\mathbf{N}}_{j;\mathbf{a}\mid \tilde{\mathbf{Z}}}(u)$ and
$\widehat{\mathbf{N}}_{j;\mathbf{b}\mid \tilde{\mathbf{Z}}}(u)$ using~\eqref{eq:partial_M_hat_a} and~\eqref{eq:partial_M_hat_b}, then
obtain their common largest eigenvalue
$\widehat{\Lambda}_{j;\mathbf{XY}\cdot \tilde{\mathbf{Z}}}^{(1)}(u)$, and set
$\widehat{\boldsymbol{\rho}}_{j;\mathbf{XY}\cdot \tilde{\mathbf{Z}}}(u)
=
\widehat{\Lambda}_{j;\mathbf{XY}\cdot \tilde{\mathbf{Z}}}^{(1)}(u)$ to be the estimated (localized) {\em partial wavelet canonical coherence}. The associated eigenvectors give the estimated (localized) {\em partial canonical vectors}.
\end{algorithmic}
\end{algorithm}

\section{Simulation Study}\label{sec:sim}

We conduct simulation studies to evaluate the empirical performance of the proposed Partial WaveCanCoh framework. The objectives are to assess whether the method can recover the direct dependence between two multivariate nonstationary processes in the presence of confounding, and to examine its finite-sample behavior under a range of representative data-generating settings. Particular attention is given to estimation accuracy, stability, and robustness to changes in the dimension and structure of the confounding process.

We simulate a joint MvLSW meta-process
$\mathbf{W}_t = \bigl(\mathbf{X}_t^\top, \mathbf{Y}_t^\top, \mathbf{Z}_t^\top\bigr)^\top$,
whose components are $\mathbf{X}_t = \bigl(X_t^{(1)}, \ldots, X_t^{(P)}\bigr)^\top$,
$\mathbf{Y}_t = \bigl(Y_t^{(1)}, \ldots, Y_t^{(Q)}\bigr)^\top$, and
$\mathbf{Z}_t = \bigl(Z_t^{(1)}, \ldots, Z_t^{(K)}\bigr)^\top$ for $t = 1, \ldots, T$, with a joint dependence structure determined by the prescribed multivariate evolutionary wavelet spectrum matrix $\mathbf{S}_j(\cdotp)$ across scales $j=1, \ldots, J=\log_2(T)$, partitioned as defined in \eqref{equ:cross-spec}. This formulation allows us to control the direct and indirect dependence among $\{\mathbf{X}_t\}$, $\{\mathbf{Y}_t\}$, and $\{\mathbf{Z}_t\}$, and therefore provides a transparent setting for evaluating the proposed methodology. Detailed specifications of the spectral constructions are given in Appendix~\ref{app:simulation}.

In the baseline setting, we take $T=1024$, $P=4$, $Q=3$, and $K=5$. We consider three scenarios. First, we study a setting in which $\{\mathbf{X}_t\}$ and $\{\mathbf{Y}_t\}$ have no direct dependence but are marginally associated through the confounding process $\{\mathbf{Z}_t\}$. In this case, the true Partial WaveCanCoh is identically zero, and the goal is to determine whether the proposed method correctly recovers the absence of direct association after adjustment for $\{\mathbf{Z}_t\}$. Since zero is a boundary point of the nonnegative coherence measure, we calibrate the estimator using an empirical null reference distribution. Specifically, we generate 5000 independent datasets from the same no-direct-dependence mechanism and define $q^{\mathrm{null}}_{0.95,j}(u)$ as the empirical 95th percentile of the resulting partial WaveCanCoh estimates at each rescaled time point $u=t/T$ and scale $j$. We further assess calibration by computing the pointwise exceedance proportion
\begin{equation}
    \frac{1}{R}\sum_{r=1}^{R}
    \mathbf{1}
    \left\{
    \widehat{\rho}^{(r)}_{j;\mathbf{XY}\cdot {\mathbf{Z}}}(u)
    >
    q^{\mathrm{null}}_{0.95,j}(u)
    \right\},
\end{equation}
where $R=1000$ denotes the  number of independent evaluation replicates. Under the null, this empirical Type I error rate should remain close to the nominal 5\% level; the corresponding pointwise Type I error curve confirms this (see  Figure \ref{fig:type1_null} in Appendix \ref{app:simulation}). 

Second, we let $\{\mathbf{X}_t\}$ and $\{\mathbf{Y}_t\}$ exhibit time-varying direct dependence while the confounding effect of $\{\mathbf{Z}_t\}$ is present. This scenario is used to evaluate whether the  estimator proposed in Section~\ref{subsec:estimation} can recover the underlying time-varying Partial WaveCanCoh. 

Third, we consider higher-dimensional confounding settings with $K=10$ and $K=20$. In these experiments, we implement the dimension-reduction procedure in Section \ref{sec:high_dim} to assess the robustness of the estimation framework in Algorithm~\ref{alg:alg_partial_hd}  when the confounding process is high-dimensional. Throughout, the data are generated using Haar wavelets. We impose nontrivial spectral structure at scale $j=1$; experiments at other scales lead to qualitatively similar conclusions (see Appendix~\ref{app:simulation} for scale-design details, with Figures~\ref{fig1:app_comparision} and~\ref{fig2:app_high_dim} for corresponding estimates).  Unless otherwise stated, all reported estimated curves are Monte Carlo averages over $R=1000$ independent replicates, and estimation uncertainty in the non-null settings is quantified using Wald-type 95\% confidence intervals.
\begin{figure}[htbp]
  \centering
  \begin{minipage}[b]{0.49\textwidth}
    \centering
    \includegraphics[width=\linewidth,height=6cm]{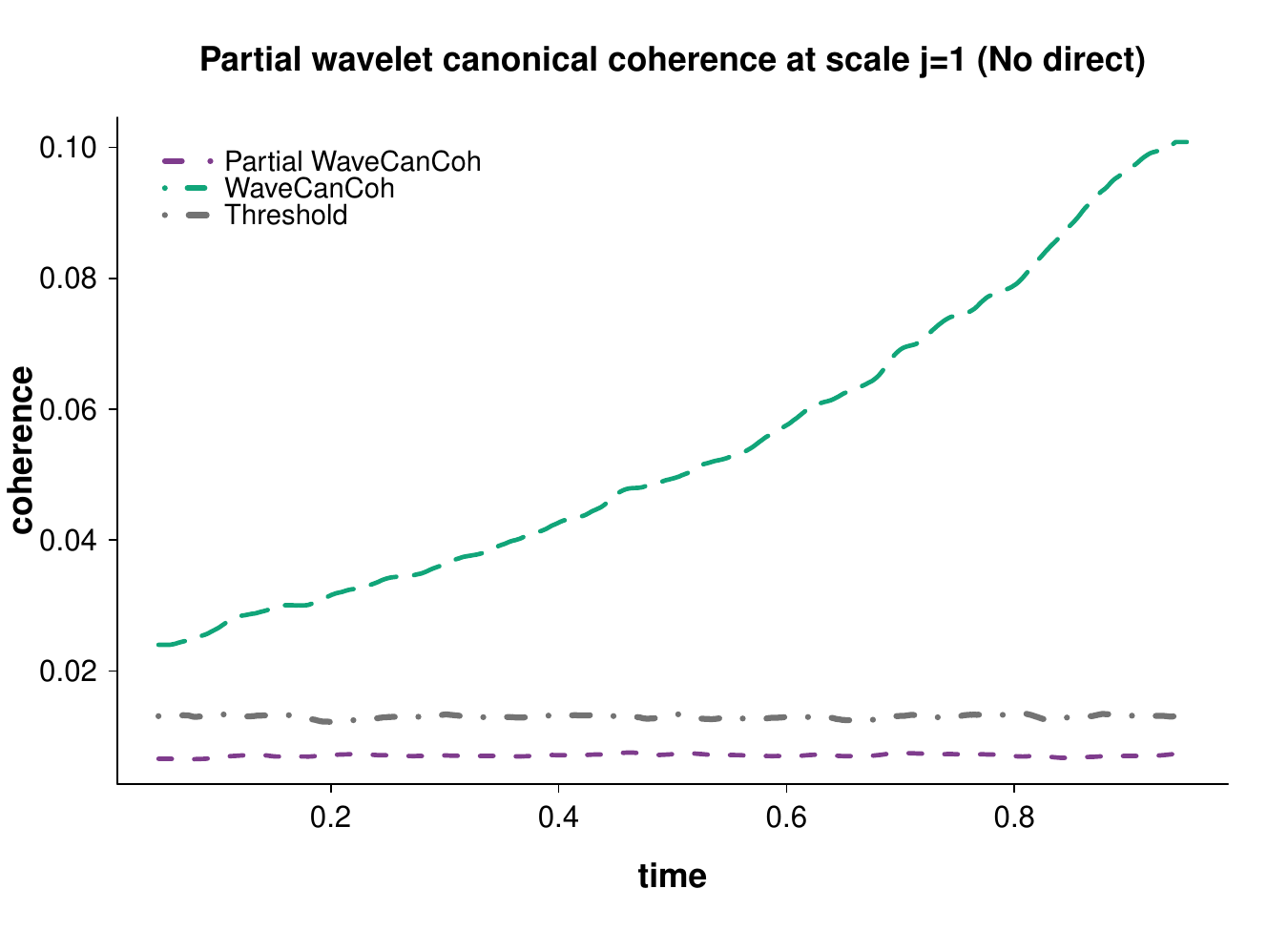}
  \end{minipage}
  \hfill
  \begin{minipage}[b]{0.49\textwidth}
    \centering
    \includegraphics[width=\linewidth,height=6cm]{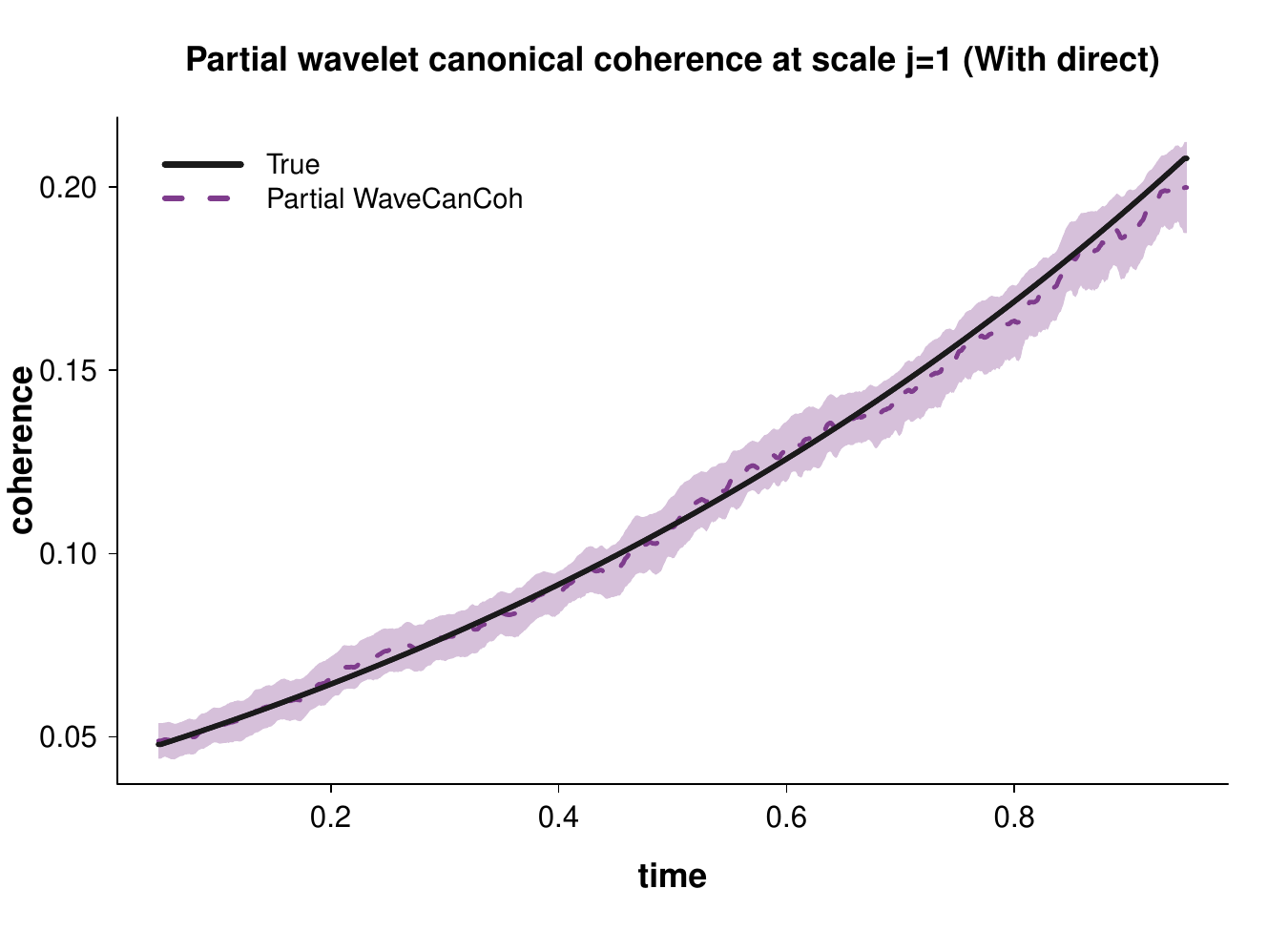}
  \end{minipage}
  \caption{Estimated partial wavelet canonical coherence (Partial WaveCanCoh) at scale $j=1$ under no direct dependence (left) and time-varying direct dependence (right). Estimated curves are Monte Carlo averages over 1000 independent replicates. The gray dashed curve in the left panel is the pointwise 95\% empirical null threshold. Shaded bands in the right panel denote Wald-type 95\% confidence intervals, with the true curve shown in black.}
  \label{fig1:simuation_comparison}
\end{figure}
\begin{figure}[htbp]
    \centering
    \includegraphics[width=0.6\linewidth]{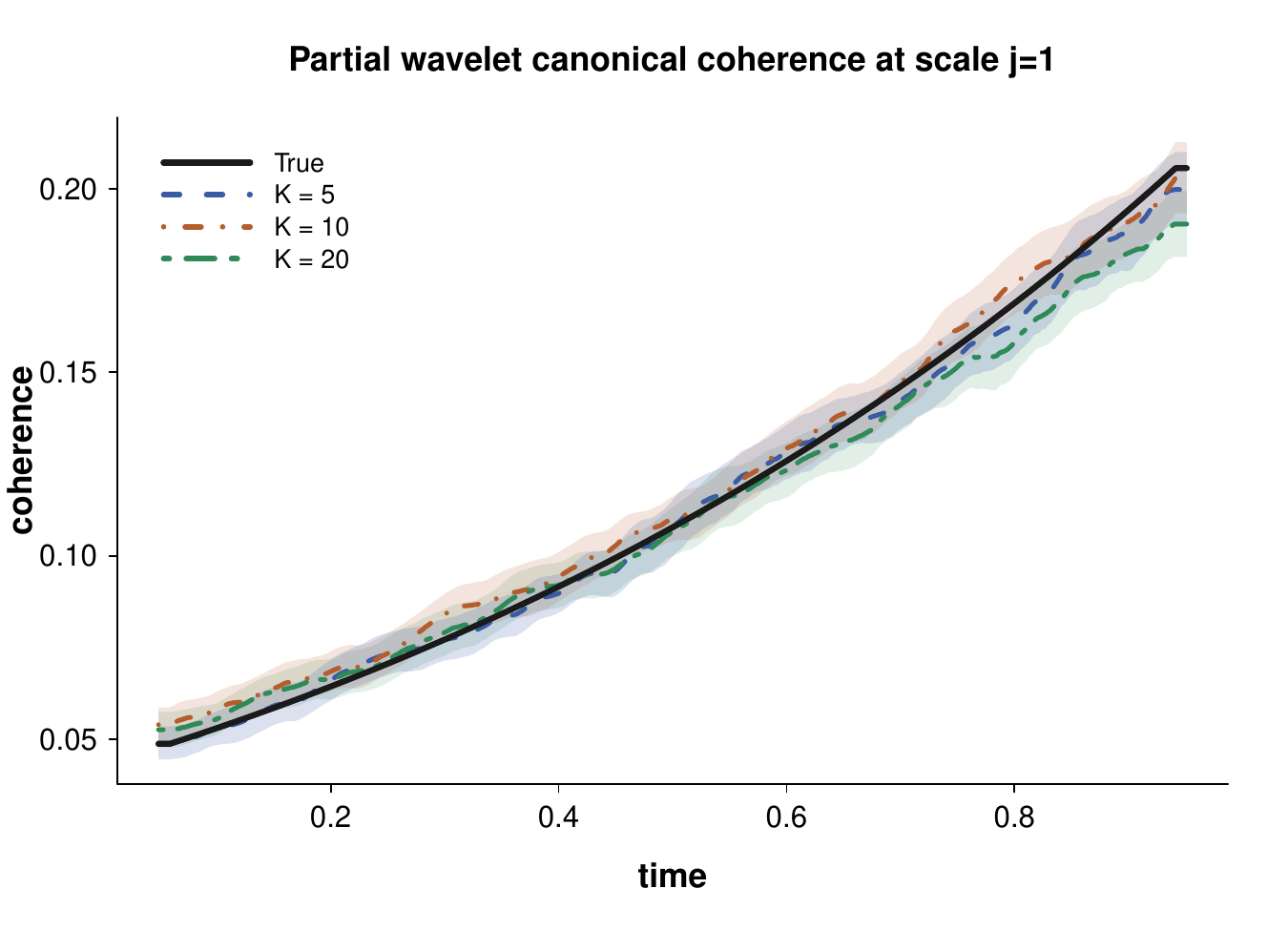}
    \caption{Partial wavelet canonical coherence at scale $j=1$ under high-dimensional confounding. The true curve is shown in black, and estimates are reported for $K=5,\, 10, \, 20$. Shaded bands denote Wald-type 95\% confidence intervals. The proposed method continues to recover the underlying time-varying direct association as the confounder dimension increases.}
    \label{fig2:sim_high_dim}
\end{figure}

Figure \ref{fig1:simuation_comparison} (left) shows that, when there is no direct coherence between $\{\mathbf{X}_t\}$ and $\{\mathbf{Y}_t\}$, the proposed Partial WaveCanCoh correctly identifies the target association as essentially zero and remains below the empirical null threshold. By contrast, the ordinary WaveCanCoh \citep{wu2025wavelet}, which does not adjust for the confounding effect of $\{\mathbf{Z}_t\}$, indicates a non-zero marginal dependence and therefore does not recover the true direct relationship. This comparison demonstrates that adjustment for confounding is necessary when the inferential target is the direct association between two multivariate nonstationary processes. Figure \ref{fig1:simuation_comparison} (right) further shows that, when time-varying direct coherence is present, the proposed estimator tracks the true target curve closely over time. Despite the confounding effect of $\{\mathbf{Z}_t\}$, the estimated partial wavelet canonical coherence remains in close agreement with the underlying direct association, indicating that the proposed procedure effectively removes the contribution of the confounding process and recovers the dependence of primary interest.

Figure \ref{fig2:sim_high_dim} summarizes the results in the presence of confounding with increasing dimension. To evaluate the proposed extension, we keep the direct coherence between $\{\mathbf{X}_t\}$ and $\{\mathbf{Y}_t\}$ fixed, while increasing the dimension of the confounding process $\{\mathbf{Z}_t\}$ and varying its dependence structure with $\{\mathbf{X}_t\}$ and $\{\mathbf{Y}_t\}$. Specifically, we consider $K=10$ and $K=20$. In each case, the high-dimensional confounding process is first reduced to dimension five 
% \textcolor{red}{why not use e.g. elbow rule?} 
using the PCA-based procedure described in Section \ref{sec:high_dim}, and the Partial WaveCanCoh is then estimated from the resulting low-dimensional representation. The results show that, despite the increased dimensionality and structural complexity of the confounding process, the proposed method continues to recover the underlying time-varying direct canonical coherence with good accuracy. These findings support the high-dimensional extension and indicate that the PCA-based reduction step provides an effective strategy for applying partial wavelet canonical coherence in the presence of high-dimensional confounders.

Overall, the simulation results show that the proposed framework can reliably isolate and estimate the direct time-varying association of interest, both in moderate- and high-dimensional confounding settings.

\section{Financial data analysis} \label{sec:data_analysis}
% To  further demonstrate the utility of Partial WaveCanCoh, we consider a real-data application in finance. Dependence across financial variables is of fundamental interest, but in practice it is often entangled with common market-wide and external confounding effects. Scale-specific Partial WaveCanCoh is particularly useful in this setting, as it allows us to isolate the direct association between two groups of financial variables after adjusting for confounding influences, and to examine how this association varies across signal components. We therefore apply the proposed method to daily financial time series, described below.
Modern equity markets are shaped by both direct linkages among asset groups and common movements induced by broader market conditions. We study daily returns of U.S. exchange-traded funds (ETFs) to examine whether two economically distinct asset groups exhibit direct dependence after external market influences have been removed. The first group consists of broad market, growth, small-cap, and technology-oriented ETFs, while the second consists of sector ETFs representing financially and macroeconomically sensitive segments of the U.S. equity market. This question is relevant for understanding cross-market transmission, diversification, and risk monitoring. Marginal dependence may overstate the direct linkage between the two groups when both are affected by common external factors, and the remaining direct dependence may vary over time and across investment horizons. Standard time-domain methods do not reveal such horizon-specific structure, while unadjusted coherence methods cannot separate direct dependence from confounding-induced association. These features make the data well suited for our Partial WaveCanCoh analysis, which targets the adjusted, time-varying and scale-specific canonical dependence between two sets of multivariate time series.

We analyze daily financial time series obtained from Yahoo Finance and accessed programmatically in R using the {\tt quantmod} package. For each asset, we use the Yahoo Finance adjusted price series and construct daily log returns. The assets are divided into three groups, denoted by $\{\mathbf{X}_t\}$, $\{\mathbf{Y}_t\}$, and $\{\mathbf{Z}_t\}$. All series are synchronized over common trading days, observations with missing values are removed, and the final sample is taken to be the most recent 1,024 common daily return observations, ranging from March 4, 2022 to April 2, 2026 (see Figure \ref{fig3:XtYt} for the realizations $\{\mathbf{X}_t\}$ and $\{\mathbf{Y}_t\}$). To examine the proposed method under increasingly high-dimensional confounding settings, we consider nested specifications of $\{\mathbf{Z}_t\}$ with dimensions $K=5, \,10, \, 20, \mbox{ and } 30$. In the higher-dimensional settings, PCA is applied to $\{\mathbf{Z}_t\}$ prior to the Partial WaveCanCoh analysis.

% \footnote{The title is not going to be appealing in an applied journal .. rewrite "Partial WaveCanCoh Analysis of Financial Time Series"}

% \footnote{This paper reads like: first, here is our method. second, this is our application. This "style" of writing is not the norm for an applied journal. Instead, it should be written this way: First, discuss the financial time series data that you want to analyze ... what do you want to test, what cross-time series relationships do you want to investigate. Second, state why this is important, what is the impact. Third, state why the current methods are not sufficient or not ideal. Finally, conduct the analysis using your method.}

% \footnote{Write a relatable story here .. something like ... We are interested, in particular, on the daily returns of the following time series data: ...  }

% \footnote{State why the other methods are not sufficient}

\begin{figure}[htbp]
  \centering
  \begin{minipage}[b]{0.49\textwidth}
    \centering
    \includegraphics[width=\linewidth,height=8cm,keepaspectratio]{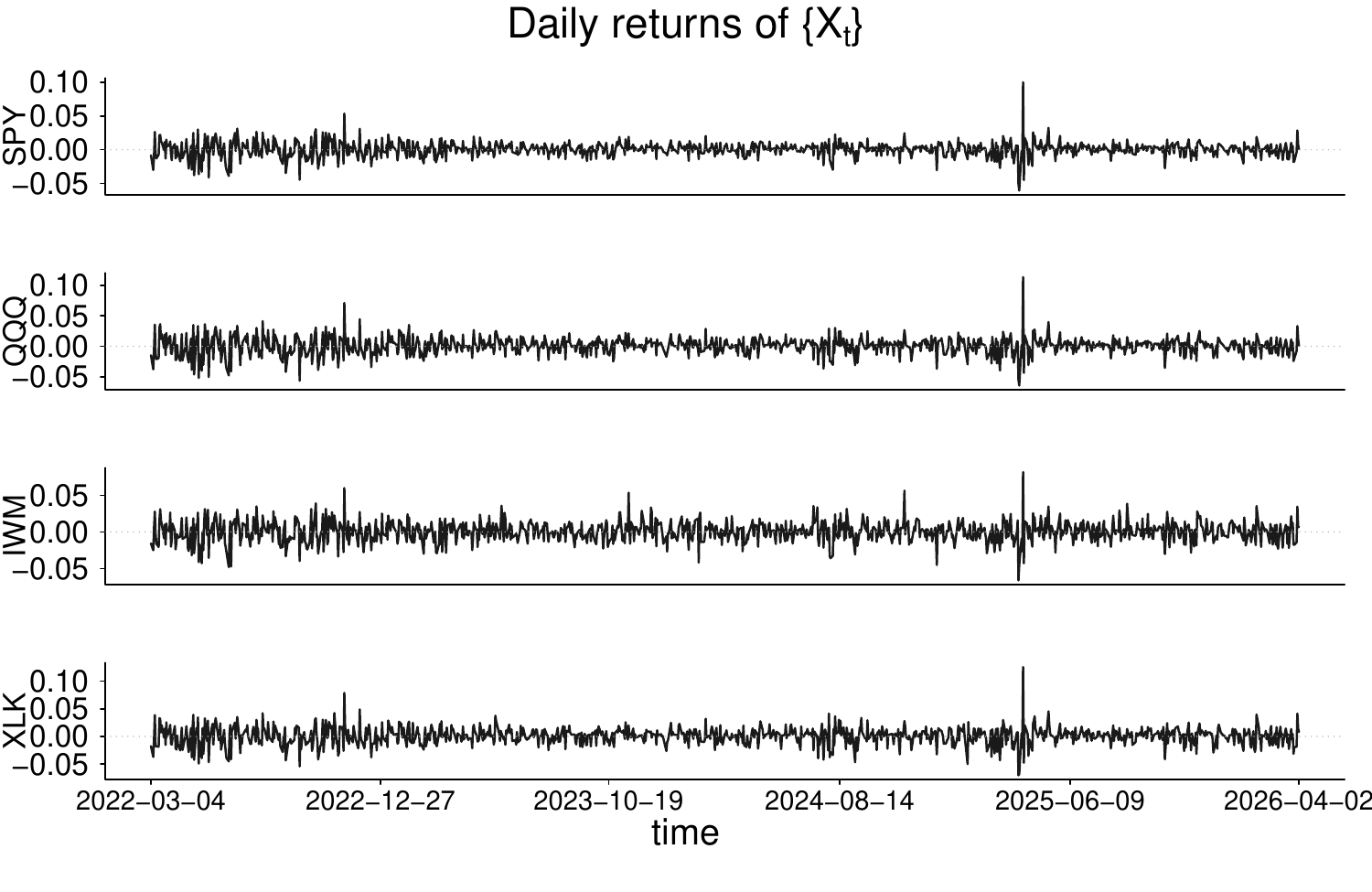}
  \end{minipage}
  \hfill
  \begin{minipage}[b]{0.49\textwidth}
    \centering
    \includegraphics[width=\linewidth,height=8cm,keepaspectratio]{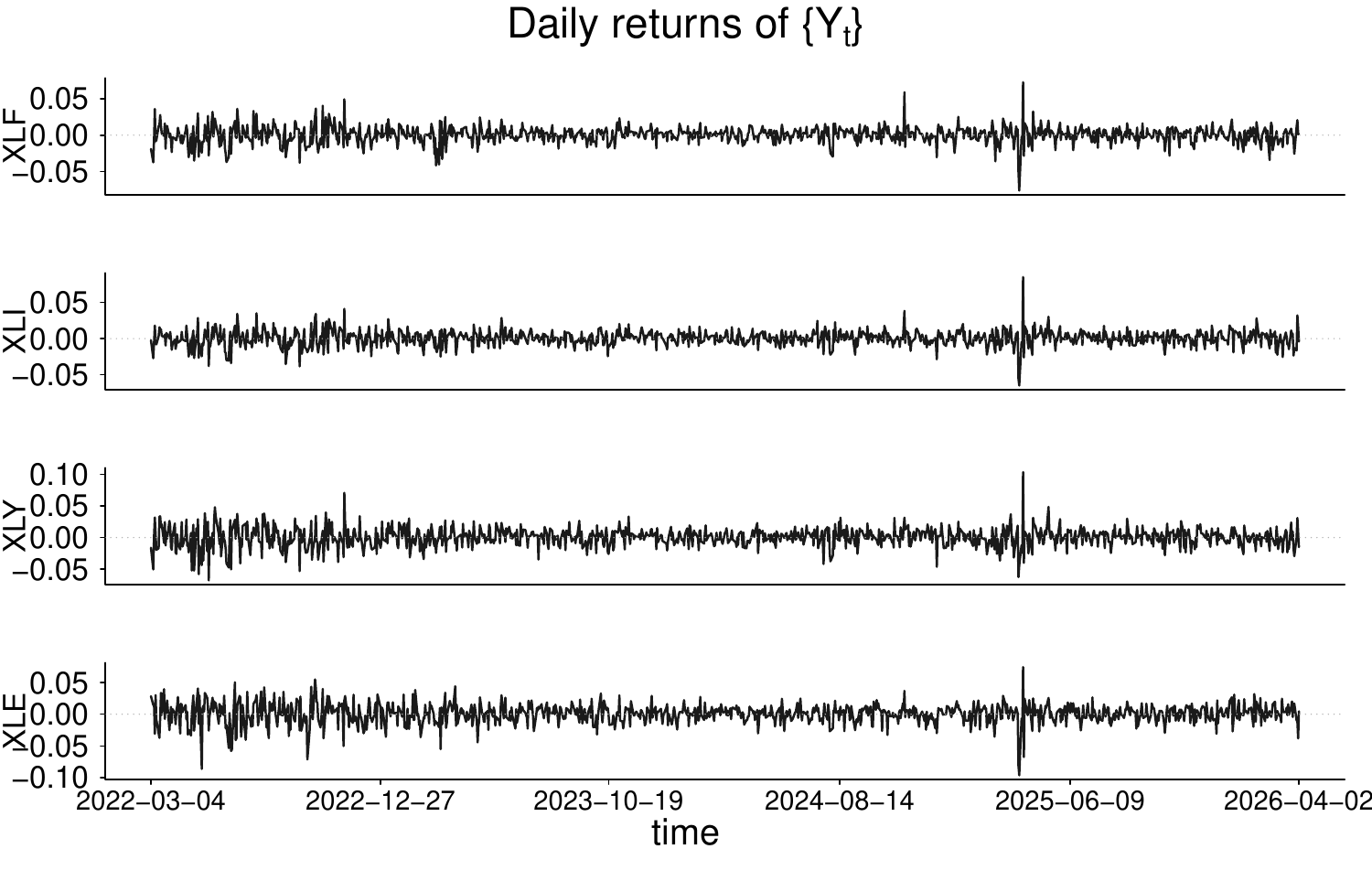}
  \end{minipage}
  \caption{Realizations of daily returns of $\{\mathbf{X}_t\}$ (left) and $\{\mathbf{Y}_t\}$ (right). The series fail the second-order stationarity test proposed in \cite{nason2013test}.}
  \label{fig3:XtYt}
\end{figure}

The sets $\{\mathbf{X}_t\}$ and $\{\mathbf{Y}_t\}$ are kept fixed and are chosen to represent two low-dimensional groups of U.S. equity market variables with distinct economic roles. Specifically, $\{\mathbf{X}_t\}$ consists of broad market, growth, small-cap, and technology-oriented exchange-traded funds (ETFs), while $\{\mathbf{Y}_t\}$ consists of sector ETFs representing financially and macroeconomically sensitive segments of the U.S. equity market. The confounding block $\{\mathbf{Z}_t\}$ is designed to capture external market influences. To preserve interpretability across dimensions, $\{\mathbf{Z}_t\}$ is constructed in a nested manner using U.S.-listed international equity ETFs, beginning with a small set of broad external market representatives and then progressively enlarging this set with additional country-level market ETFs. This yields increasingly rich confounding specifications while maintaining a coherent interpretation of $\{\mathbf{Z}_t\}$ as an external market block. %In the empirical analysis, 

In particular, we take $\{\mathbf{X}_t\}=(\mathrm{SPY},\mathrm{QQQ},\mathrm{IWM},\mathrm{XLK})^\top$ and $\{\mathbf{Y}_t\}=(\mathrm{XLF},\mathrm{XLI},\mathrm{XLY},\mathrm{XLE})^\top$, while the smallest confounding specification starts from $\{\mathbf{Z}_t\}=(\mathrm{EFA},\mathrm{EEM},\mathrm{EZU},\mathrm{EWJ},\mathrm{FXI})^\top$ and is then enlarged sequentially. The full composition of the sets $\{\mathbf{X}_t\}$, $\{\mathbf{Y}_t\}$, and the nested specifications of $\{\mathbf{Z}_t\}$ are reported in Tables \ref{tab:xy_groups} and \ref{tab:z_groups}, Appendix \ref{app:data_analysis}.
\begin{figure}[htbp]
    \centering
    
    \begin{minipage}[b]{0.48\textwidth}
        \centering
        \includegraphics[width=\linewidth]{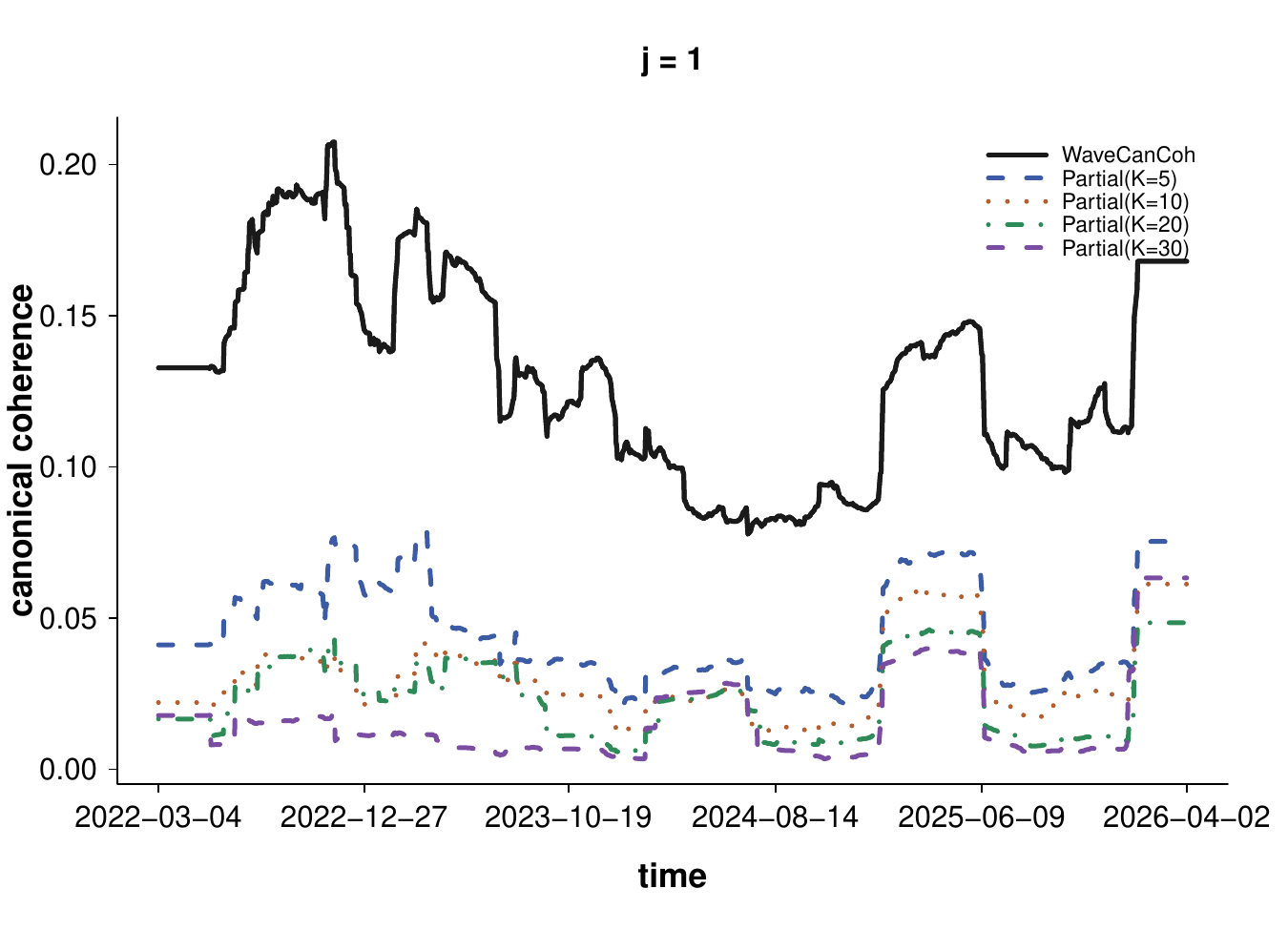}
    \end{minipage}
    \hfill
    \begin{minipage}[b]{0.48\textwidth}
        \centering
        \includegraphics[width=\linewidth]{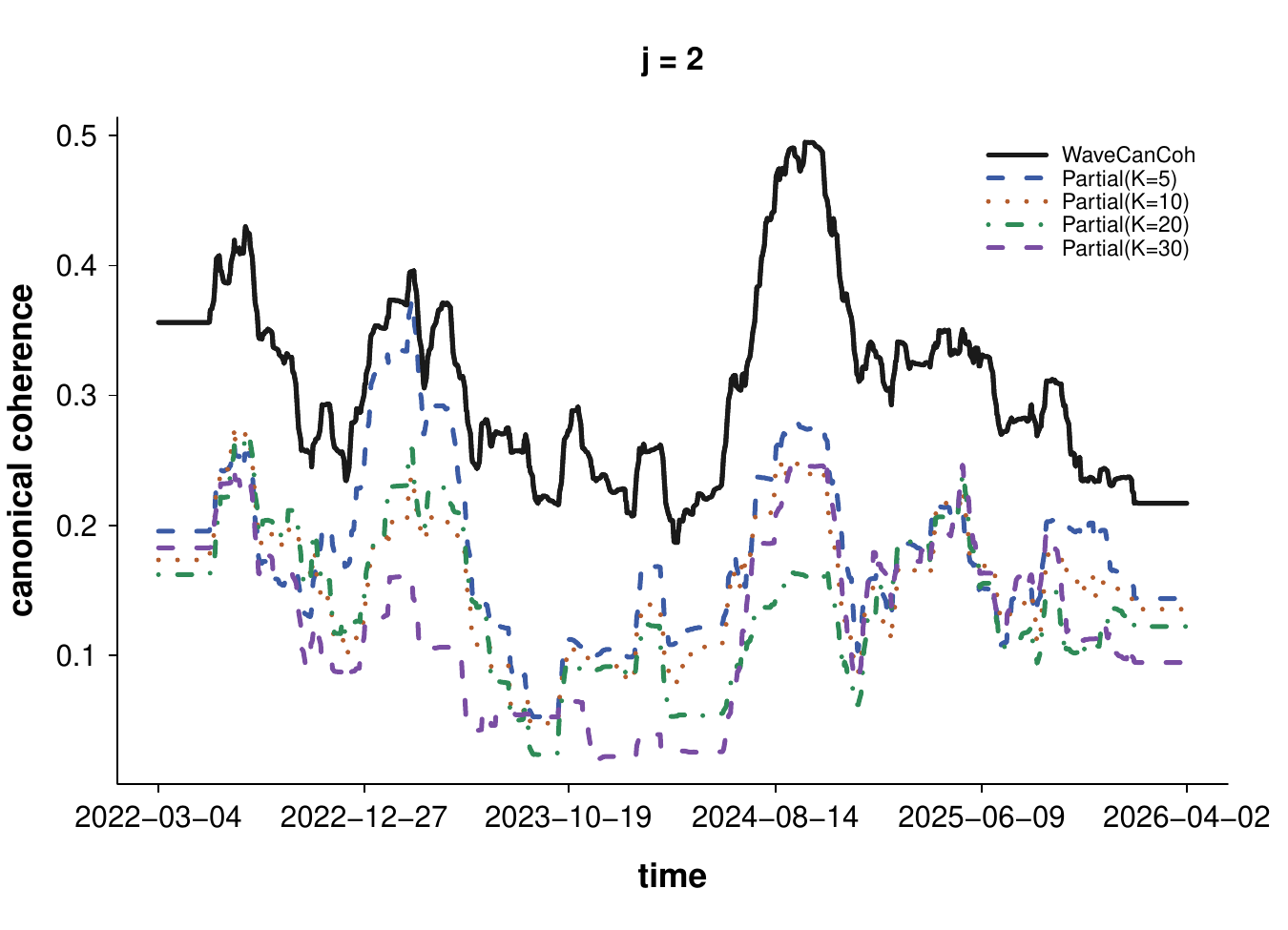}
    \end{minipage}
    
    \vspace{0.5em}
    
    \begin{minipage}[b]{0.48\textwidth}
        \centering
        \includegraphics[width=\linewidth]{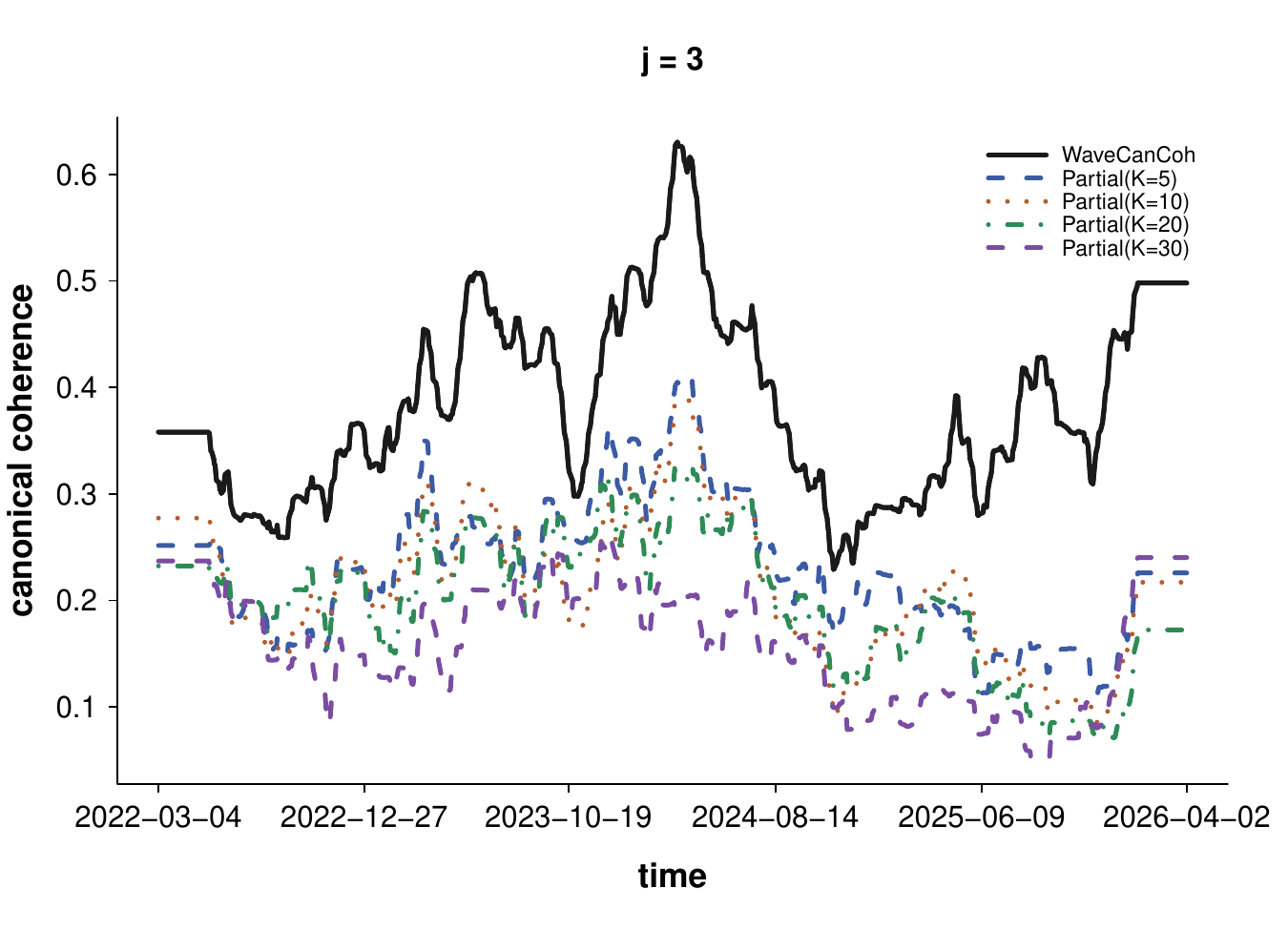}
    \end{minipage}
    \hfill
    \begin{minipage}[b]{0.48\textwidth}
        \centering
        \includegraphics[width=\linewidth]{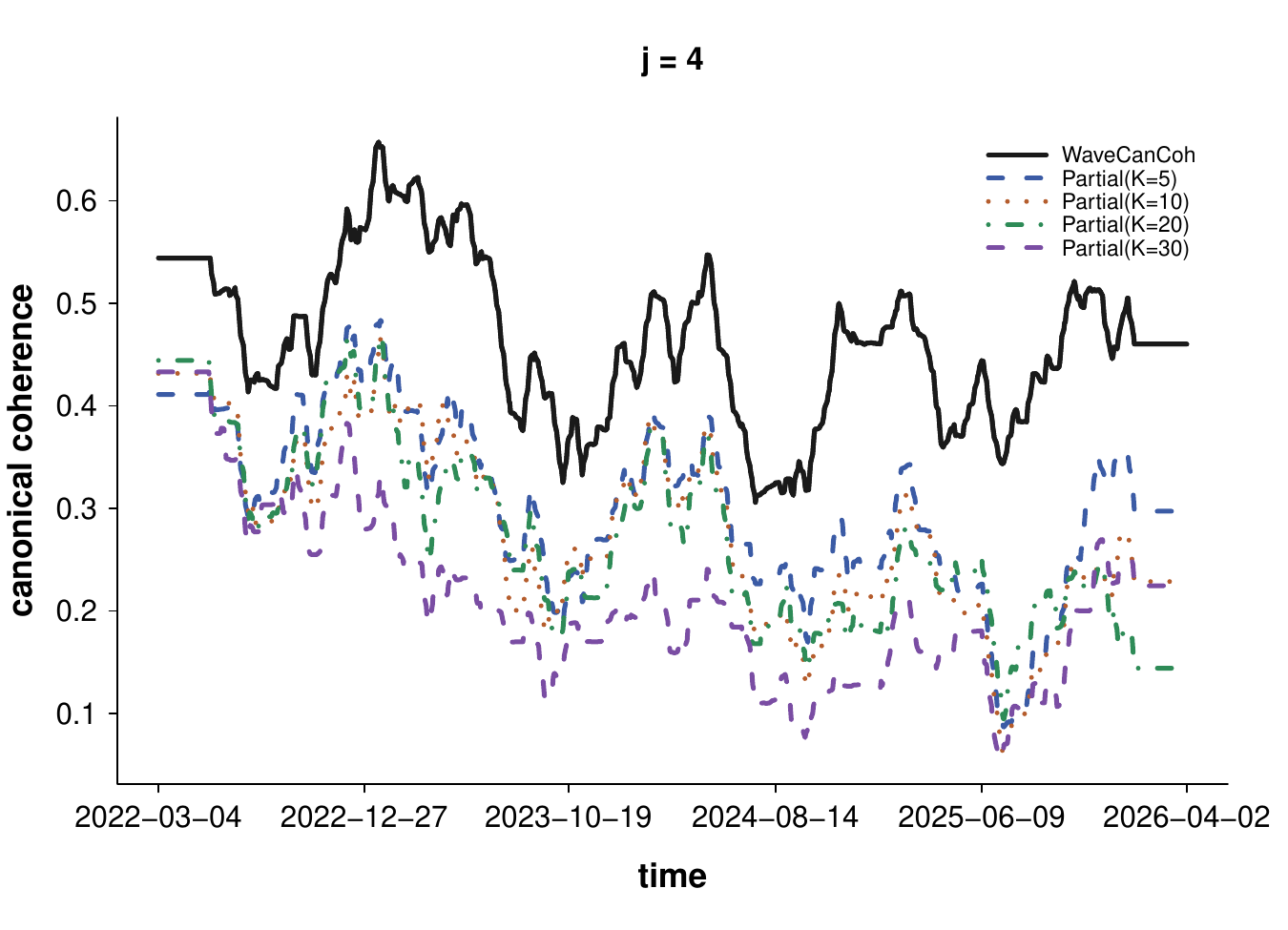}
    \end{minipage}
    
    \caption{Estimates of wavelet-based canonical coherence (WaveCanCoh) and partial canonical coherence (Partial WaveCanCoh) across scales $j=1,2,3,4$ for daily returns of U.S. equities. Each panel shows the estimated wavelet canonical coherence at a given scale for the unadjusted analysis and for partial adjustment with confounder dimensions $K=5,10,20,$ and $30$. For $K>5$, the confounding variables are reduced to five principal components before estimating Partial WaveCanCoh.}
    \label{fig:real_data}
\end{figure}

Figure \ref{fig:real_data} shows clear time-varying and scale-dependent structure in the estimated direct association between $\{\mathbf{X}_t\}$ and $\{\mathbf{Y}_t\}$ across four scales. At each scale, the canonical coherence varies substantially over time, while the coherence patterns differ noticeably across scales, indicating that the direct relationship between the two groups is concentrated in different frequency components at different times. Under the dyadic wavelet decomposition, the finer scales correspond roughly to shorter-horizon dynamics; for example, $j=1$ is associated with approximately bi-daily variation, whereas $j=2$ reflects roughly weekly movements while scales $j=3, \, 4$ correspond to monthly and seasonal dynamics, respectively. Interestingly, periods displaying weak dependence (e.g., time interval between 2023 and 2024) at the finest scale ($j=1$) appear actively dependent at coarser scales (e.g., $j=3$). This highlights a key advantage of our proposed wavelet-based framework: it does not merely provide a single global summary of dependence, but identifies which signal components are directly associated and when those associations are most pronounced. Being able to couple the time--scale (horizon) dynamics suggests from a financial perspective that the direct interaction between the broad equity block $\{\mathrm{SPY},\mathrm{QQQ},\mathrm{IWM},\mathrm{XLK}\}$ and the sector block $\{\mathrm{XLF},\mathrm{XLI},\mathrm{XLY},\mathrm{XLE}\}$ is not homogeneous across investment horizons: some episodes are driven primarily by shorter-horizon transmission of market sentiment and sector rotation, whereas others appear at lower-frequency scales and are more consistent with broader macroeconomic or business-cycle comovement.

Across all four scales, the unadjusted WaveCanCoh estimates are generally larger than the corresponding Partial WaveCanCoh ones, indicating that a substantial part of the apparent association between $\{\mathbf{X}_t\}$ and $\{\mathbf{Y}_t\}$ is in fact attributable to common external market effects captured by $\{\mathbf{Z}_t\}$. Given that $\{\mathbf{Z}_t\}$ consists of international equity ETFs, this reduction is consistent with the presence of global market factors that jointly affect both the core U.S. equity indexes and the sector-specific U.S. equity block. As the dimension of the confounding block increases from $K=5$ to $K=30$, the partial canonical coherence curves overall tend to shift downward, suggesting that richer confounding specifications remove additional shared variation and yield a more conservative estimate of the underlying direct association. This overall pattern provides empirical evidence that confounding effects are present and that partial adjustment is necessary for a more reliable characterization of the direct cross-group dependence. The estimates are not pointwise monotone in $K$, which is not unexpected, as partial canonical coherence is not generally monotone in the conditioning set. Moreover, for $K>5$, the confounding block is first reduced to a five-dimensional PCA-based representation, so different values of $K$ do not form a strictly nested sequence of conditioning variables. Hence, enlarging the raw confounding set need not induce a pointwise decrease in the estimated partial coherence. The salient empirical features are the overall downward shift relative to the unadjusted WaveCanCoh, together with the retention of clear time-varying and scale-dependent structure after adjustment.

\section{Conclusions} \label{sec:conc}
In this paper, we developed Partial Wavelet Canonical Coherence, a framework for measuring the direct association between two sets of multivariate nonstationary time series in the presence of confounding processes. The proposed methodology extends partial canonical correlation to the wavelet domain and, to our knowledge, provides the first scale-specific and time-varying measure of partial set-to-set dependence for multivariate time series. By operating within the multivariate locally stationary wavelet framework, it accommodates nonstationarity naturally and allows direct association to be localized both across scales and over time. The simulation and empirical results show that the method can distinguish direct dependence from confounded marginal association and can remain effective even when the confounding block is high-dimensional through suitable dimension reduction.

Several directions merit further investigation. On the applied side, the proposed framework may be especially useful in biomedical and neuroscientific settings \citep{embleton2022multiscale}, where one seeks direct interactions between multichannel systems after adjusting for common background activity, and in climate or environmental studies, where large-scale external drivers may mask localized cross-system dependence. On the methodological side, an important next step is to integrate Partial WaveCanCoh with richer dimension-reduction strategies for dependent and nonstationary data, including dynamic factor models, dynamic principal components, and wavelet-domain reduction methods. It would also be of interest to develop corresponding inferential theory and regularized variants suited to increasingly high-dimensional settings, thereby further broadening the scope of the proposed framework in modern multivariate time series analysis. 

\section{Conflicts of interest}
The authors declare that they have no competing interests.

\section{Funding}
MIK gratefully acknowledges support from  EPSRC NeST Programme Grant EP/X002195/1.

\section{Data availability}
The data underlying this article are available in {\tt quantmod} package in R.

%UNCOMMENT THE BELOW TWO LINES IN CASE YOU NEED NUMBERED FORMAT.
%\bibliographystyle{oup-plain}
%\bibliography{reference}

%UNCOMMENT THE BELOW TWO LINES IN CASE YOU NEED AUTHOR YEAR FORMAT.
\bibliographystyle{oup-abbrvnat}
%\bibliography{reference}

\bibliography{references}

% \clearpage

% \begin{center}
%   {\Large\bfseries Appendix}
% \end{center}
% \vspace{1em}

% \setcounter{section}{0}
% \setcounter{subsection}{0}
% \setcounter{equation}{0}
% \setcounter{figure}{0}
% \setcounter{table}{0}

% \renewcommand{\thesection}{\Alph{section}}
% \renewcommand{\thesubsection}{\Alph{section}.\arabic{subsection}}
% \renewcommand{\theequation}{\Alph{section}.\arabic{equation}}
% \renewcommand{\thefigure}{\Alph{section}.\arabic{figure}}
% \renewcommand{\thetable}{\Alph{section}.\arabic{table}}
\begin{appendices}

 \section{Theoretical Results}
\label{app:proof}
\textbf{Proof for Proposition 1} 
\begin{proof}
The proof follows a similar variational argument to that used for WaveCanCoh in \cite{wu2025wavelet}, with the ordinary local wavelet spectral matrices replaced by their partial counterparts, as follows. %and we include it here for completeness.% in the partial setting.

Suppose that $\{\mathbf{X}_t\}$, $\{\mathbf{Y}_t\}$, and $\{\mathbf{Z}_t\}$ are MvLSW time series, and for each scale $j$ and rescaled time $u$, the goal is to find the vectors $\mathbf{a}_j(u)$ and $\mathbf{b}_j(u)$ that maximize the partial canonical coherence between $\{\mathbf{X}_t\}$ and $\{\mathbf{Y}_t\}$ given the information in the process $\{\mathbf{Z}_t\}$,
\begin{align*}
    \boldsymbol{\rho}_{j;\mathbf{XY\cdot Z}}(u)
    =
    \left\{
    \mathbf{a}_j^{\top}(u)\mathbf{S}_{j;\mathbf{XY\cdot Z}}(u)\mathbf{b}_j(u)
    \right\}^2
\end{align*}
subject to the normalization constraints
\begin{align*}
    \mathbf{a}_j^{\top}(u)\mathbf{S}_{j;\mathbf{XX\cdot Z}}(u)\mathbf{a}_j(u)=1,
    \quad
    \mathbf{b}_j^{\top}(u)\mathbf{S}_{j;\mathbf{YY\cdot Z}}(u)\mathbf{b}_j(u)=1.
\end{align*}
Assuming that the scale $j$ and rescaled time $u$ are fixed, we suppress them for notational simplicity. To solve the above optimization problem, we consider the Lagrangian
\begin{align*}
\mathcal{L}(\mathbf{a},\mathbf{b},\lambda_1,\lambda_2)
&=
\frac{1}{2}\left(\mathbf{a}^{\top}\mathbf{S}_{\mathbf{XY\cdot Z}}\mathbf{b}\right)^2
-\frac{\lambda_1}{2}\left(\mathbf{a}^{\top}\mathbf{S}_{\mathbf{XX\cdot Z}}\mathbf{a}-1\right)
-\frac{\lambda_2}{2}\left(\mathbf{b}^{\top}\mathbf{S}_{\mathbf{YY\cdot Z}}\mathbf{b}-1\right).
\end{align*}
Differentiating the Lagrangian with respect to $\mathbf{a}$ and $\mathbf{b}$, respectively, and setting the derivatives equal to zero gives
\begin{align*}
\frac{\partial \mathcal{L}}{\partial \mathbf{a}}
&=
\left(\mathbf{S}_{\mathbf{XY\cdot Z}}\mathbf{b}\right)
\left(\mathbf{S}_{\mathbf{XY\cdot Z}}\mathbf{b}\right)^{\top}\mathbf{a}
-\lambda_1\mathbf{S}_{\mathbf{XX\cdot Z}}\mathbf{a}
=\mathbf{0} \\
&\Rightarrow
\left(\mathbf{S}_{\mathbf{XY\cdot Z}}\mathbf{b}\right)
\left(\mathbf{S}_{\mathbf{XY\cdot Z}}\mathbf{b}\right)^{\top}\mathbf{a}
=
\lambda_1\mathbf{S}_{\mathbf{XX\cdot Z}}\mathbf{a},
\qquad \mathrm{(i)}
\\[0.3em]
\frac{\partial \mathcal{L}}{\partial \mathbf{b}}
&=
\left(\mathbf{S}_{\mathbf{YX\cdot Z}}\mathbf{a}\right)
\left(\mathbf{S}_{\mathbf{YX\cdot Z}}\mathbf{a}\right)^{\top}\mathbf{b}
-\lambda_2\mathbf{S}_{\mathbf{YY\cdot Z}}\mathbf{b}
=\mathbf{0} \\
&\Rightarrow
\left(\mathbf{S}_{\mathbf{YX\cdot Z}}\mathbf{a}\right)
\left(\mathbf{S}_{\mathbf{YX\cdot Z}}\mathbf{a}\right)^{\top}\mathbf{b}
=
\lambda_2\mathbf{S}_{\mathbf{YY\cdot Z}}\mathbf{b},
\qquad \mathrm{(ii)}
\end{align*}
where $\mathbf{S}_{\mathbf{YX\cdot Z}}=\mathbf{S}_{\mathbf{XY\cdot Z}}^{\top}$.

Using $\mathrm{(i)}$ and $\mathrm{(ii)}$, and noting that $\mathbf{a}^{\top}\mathbf{S}_{\mathbf{XY\cdot Z}}\mathbf{b}$ is a real-valued scalar, we further obtain
\begin{align*}
(\mathbf{a}^{\top}\mathbf{S}_{\mathbf{XY\cdot Z}}\mathbf{b})
(\mathbf{a}^{\top}\mathbf{S}_{\mathbf{XY\cdot Z}}\mathbf{b})
&=
\lambda_1\mathbf{a}^{\top}\mathbf{S}_{\mathbf{XX\cdot Z}}\mathbf{a}, \\
(\mathbf{b}^{\top}\mathbf{S}_{\mathbf{YX\cdot Z}}\mathbf{a})
(\mathbf{a}^{\top}\mathbf{S}_{\mathbf{XY\cdot Z}}\mathbf{b})
&=
\lambda_2\mathbf{b}^{\top}\mathbf{S}_{\mathbf{YY\cdot Z}}\mathbf{b}.
\end{align*}
Recalling the constraints
\begin{align*}
\mathbf{a}^{\top}\mathbf{S}_{\mathbf{XX\cdot Z}}\mathbf{a}=1,
\quad
\mathbf{b}^{\top}\mathbf{S}_{\mathbf{YY\cdot Z}}\mathbf{b}=1,
\end{align*}
it follows immediately that
\begin{align*}
\lambda_1=\lambda_2=
\left(\mathbf{a}^{\top}\mathbf{S}_{\mathbf{XY\cdot Z}}\mathbf{b}\right)^2
=: \lambda.
\end{align*}

Substituting this back into $\mathrm{(i)}$ and $\mathrm{(ii)}$, and assuming $\lambda\neq 0$, yields
\begin{align*}
\mathbf{a}
=
\frac{1}{\sqrt{\lambda}}
\mathbf{S}_{\mathbf{XX\cdot Z}}^{-1}\mathbf{S}_{\mathbf{XY\cdot Z}}\mathbf{b}, \quad
\mathbf{b}
=
\frac{1}{\sqrt{\lambda}}
\mathbf{S}_{\mathbf{YY\cdot Z}}^{-1}\mathbf{S}_{\mathbf{YX\cdot Z}}\mathbf{a}.
\end{align*}
Using these expressions into $\mathrm{(i)}$ and $\mathrm{(ii)}$, respectively, gives
\begin{align*}
\mathbf{S}_{\mathbf{XY\cdot Z}}\mathbf{S}_{\mathbf{YY\cdot Z}}^{-1}\mathbf{S}_{\mathbf{YX\cdot Z}}\mathbf{a}
=
\lambda\,\mathbf{S}_{\mathbf{XX\cdot Z}}\mathbf{a}, \quad
\mathbf{S}_{\mathbf{YX\cdot Z}}\mathbf{S}_{\mathbf{XX\cdot Z}}^{-1}\mathbf{S}_{\mathbf{XY\cdot Z}}\mathbf{b}
=
\lambda\,\mathbf{S}_{\mathbf{YY\cdot Z}}\mathbf{b},
\end{align*}
or, equivalently,
\begin{align*}
\mathbf{S}_{\mathbf{XX\cdot Z}}^{-1}\mathbf{S}_{\mathbf{XY\cdot Z}}\mathbf{S}_{\mathbf{YY\cdot Z}}^{-1}\mathbf{S}_{\mathbf{YX\cdot Z}}\mathbf{a}
=
\lambda\,\mathbf{a}, \quad
\mathbf{S}_{\mathbf{YY\cdot Z}}^{-1}\mathbf{S}_{\mathbf{YX\cdot Z}}\mathbf{S}_{\mathbf{XX\cdot Z}}^{-1}\mathbf{S}_{\mathbf{XY\cdot Z}}\mathbf{b}
=
\lambda\,\mathbf{b}.
\end{align*}
Hence,
$\lambda=\left(\mathbf{a}^{\top}\mathbf{S}_{\mathbf{XY\cdot Z}}\mathbf{b}\right)^2
$ is an eigenvalue of both matrices
\begin{align*}
\mathbf{N}_{\mathbf{a}\mid \mathbf{Z}}:=\mathbf{S}_{\mathbf{XX\cdot Z}}^{-1}\mathbf{S}_{\mathbf{XY\cdot Z}}\mathbf{S}_{\mathbf{YY\cdot Z}}^{-1}\mathbf{S}_{\mathbf{YX\cdot Z}} \text{ and } 
\mathbf{N}_{\mathbf{b}\mid \mathbf{Z}}:=\mathbf{S}_{\mathbf{YY\cdot Z}}^{-1}\mathbf{S}_{\mathbf{YX\cdot Z}}\mathbf{S}_{\mathbf{XX\cdot Z}}^{-1}\mathbf{S}_{\mathbf{XY\cdot Z}},
\end{align*}
with corresponding eigenvectors $\mathbf{a}$ and $\mathbf{b}$, respectively. Therefore, the partial wavelet canonical coherence as defined in Definition~\ref{def:partialwcc},
\begin{align*}
\boldsymbol{\rho}_{j;\mathbf{XY\cdot Z}}(u)
=
\max_{\mathbf{a}_j(u), \mathbf{b}_j(u)}
\left\{
\mathbf{a}_j^{\top}(u)\mathbf{S}_{j;\mathbf{XY\cdot Z}}(u)\mathbf{b}_j(u)
\right\}^2
\end{align*}
is in fact given by the largest eigenvalue of the above matrices $\mathbf{N}_{j;\mathbf{a}\mid \mathbf{Z}}(u)$ and $\mathbf{N}_{j;\mathbf{b}\mid \mathbf{Z}}(u)$. The case $\lambda=0$ corresponds to zero partial canonical coherence between the $\mathbf{X}$- and $\mathbf{Y}$- processes after adjusting for $\mathbf{Z}$, which is not of primary interest here. Finally, we note that the above derivation holds for every scale $j$ and rescaled time $u$.
\end{proof}

\textbf{Proof for Proposition 2}
\begin{proof}
For ease of notation, in what follows we drop the spectral dependence on rescaled time $u$. We aim to establish the consistency of the matrix estimators 
\begin{align*}
\widehat{\mathbf{N}}_{j;\mathbf{a}\mid \mathbf{Z}}
&=
\widehat{\mathbf{S}}_{j;\mathbf{XX\cdot Z}}^{-1}
\widehat{\mathbf{S}}_{j;\mathbf{XY\cdot Z}}
\widehat{\mathbf{S}}_{j;\mathbf{YY\cdot Z}}^{-1}
\widehat{\mathbf{S}}_{j;\mathbf{YX\cdot Z}}, \\
\widehat{\mathbf{N}}_{j;\mathbf{b}\mid \mathbf{Z}}
&=
\widehat{\mathbf{S}}_{j;\mathbf{YY\cdot Z}}^{-1}
\widehat{\mathbf{S}}_{j;\mathbf{YX\cdot Z}}
\widehat{\mathbf{S}}_{j;\mathbf{XX\cdot Z}}^{-1}
\widehat{\mathbf{S}}_{j;\mathbf{XY\cdot Z}}.
\end{align*}
Under the uniform autocovariance absolute summability assumption, \cite{ombao2014} have shown that the corrected smoothed periodogram estimator of the meta-process LWS in~\eqref{eq:joint_lws_est} is consistent for the true spectrum. Therefore, as $T\to\infty$ and the smoothing parameter $M\to\infty$ with $M/T\to 0$, we have for the component relevant block local wavelet spectral matrices
\begin{align*}
\widehat{\mathbf{S}}_{j;\mathbf{AB}} \xrightarrow{P} \mathbf{S}_{j;\mathbf{AB}},
\end{align*}
for all pairs $\mathbf{AB}\in\{\mathbf{XX},\mathbf{XY},\mathbf{YX},\mathbf{YY},\mathbf{XZ},\mathbf{ZX},\mathbf{YZ},\mathbf{ZY},\mathbf{ZZ}\}$.
Since the partial LWS matrices are obtained from these block matrices through matrix multiplication, subtraction, and inversion (see equations~\eqref{eq:partial_lws_xx},~\eqref{eq:partial_lws_yy} and~\eqref{eq:partial_lws_xy}), it follows by the continuous mapping theorem and Slutsky's theorem \citep{billingsley1999convergence} that
\begin{align*}
\widehat{\mathbf{S}}_{j;\mathbf{XX\cdot Z}} &\xrightarrow{P} \mathbf{S}_{j;\mathbf{XX\cdot Z}}, \qquad
\widehat{\mathbf{S}}_{j;\mathbf{YY\cdot Z}} \xrightarrow{P} \mathbf{S}_{j;\mathbf{YY\cdot Z}}, \\
\widehat{\mathbf{S}}_{j;\mathbf{XY\cdot Z}} &\xrightarrow{P} \mathbf{S}_{j;\mathbf{XY\cdot Z}}, \qquad
\widehat{\mathbf{S}}_{j;\mathbf{YX\cdot Z}} \xrightarrow{P} \mathbf{S}_{j;\mathbf{YX\cdot Z}},
\end{align*}
assuming that the required spectral matrices are nonsingular.

Applying the continuous mapping theorem once more gives
\begin{align*}
\widehat{\mathbf{S}}_{j;\mathbf{XX\cdot Z}}^{-1}
\xrightarrow{P}
\mathbf{S}_{j;\mathbf{XX\cdot Z}}^{-1},
\qquad
\widehat{\mathbf{S}}_{j;\mathbf{YY\cdot Z}}^{-1}
\xrightarrow{P}
\mathbf{S}_{j;\mathbf{YY\cdot Z}}^{-1}.
\end{align*}
Hence, by the estimator construction in equations~\eqref{eq:partial_M_hat_a} and~\eqref{eq:partial_M_hat_b} coupled with Slutsky's theorem,
\begin{align*}
\widehat{\mathbf{N}}_{j;\mathbf{a}\mid \mathbf{Z}}
\xrightarrow{P}
\mathbf{N}_{j;\mathbf{a}\mid \mathbf{Z}},
\qquad
\widehat{\mathbf{N}}_{j;\mathbf{b}\mid \mathbf{Z}}
\xrightarrow{P}
\mathbf{N}_{j;\mathbf{b}\mid \mathbf{Z}}.
\end{align*}
Consequently, the estimated partial wavelet canonical coherence, together with the associated canonical vectors obtained from the largest eigenvalue and corresponding eigenvectors of $\widehat{\mathbf{N}}_{j;\mathbf{a}\mid \mathbf{Z}}$ and $\widehat{\mathbf{N}}_{j;\mathbf{b}\mid \mathbf{Z}}$, converge in probability to their population counterparts, by arguments analogous to those in \cite{knight2024adaptive}.
\end{proof}

\section{Supplementary simulation details and results}
\label{app:simulation}
We simulate a multivariate locally stationary wavelet (MvLSW) process
\[
\mathbf{W}_t=
\bigl(\mathbf{X}_t^\top,\mathbf{Y}_t^\top,\mathbf{Z}_t^\top\bigr)^\top,
\qquad t=1,\dots,T,
\]
with $T=1024$, where $\mathbf{X}_t\in\mathbb{R}^4$, $\mathbf{Y}_t\in\mathbb{R}^3$, and $\mathbf{Z}_t\in\mathbb{R}^K$, stemming from
\[
\mathbf{S}_{j;\mathbf{W}\mathbf{W}}(u), \qquad u=t/T\in(0,1),
\]
which denotes the local wavelet spectral matrix of the process $\{\mathbf{W}_t\}$. In all simulations, the spectrum is non-zero at one scale and we consider two such designs, with active scale $j=1$ and $j=4$, respectively. In both cases, the spectral matrix is partitioned as
\[
\mathbf{S}_{j;\mathbf{W}\mathbf{W}}(u)=
\begin{bmatrix}
\mathbf{S}_{j;\mathbf{XX}}(u) & \mathbf{S}_{j;\mathbf{XY}}(u) & \mathbf{S}_{j;\mathbf{XZ}}(u)\\
\mathbf{S}_{j;\mathbf{YX}}(u) & \mathbf{S}_{j;\mathbf{YY}}(u) & \mathbf{S}_{j;\mathbf{YZ}}(u)\\
\mathbf{S}_{j;\mathbf{ZX}}(u) & \mathbf{S}_{j;\mathbf{ZY}}(u) & \mathbf{S}_{j;\mathbf{ZZ}}(u)
\end{bmatrix}.
\]
The direct association between $\{\mathbf{X}_t\}$ and $\{\mathbf{Y}_t\}$ is generated through
\[
\mathbf{S}_{j;\mathbf{XY}}(u)
=
\mathbf{S}_{j;\mathbf{XZ}}(u)\,
\mathbf{S}_{j;\mathbf{ZZ}}(u)^{-1}\,
\mathbf{S}_{j;\mathbf{ZY}}(u)
+
\boldsymbol{\Delta}_{XY}^{(j)}(u) ,
\]
so that the observed dependence between $\{\mathbf{X}_t\}$ and $\{\mathbf{Y}_t\}$ contains both an induced component through the confounders and, when $\boldsymbol{\Delta}_{XY}^{(j)}(u)\neq \mathbf{0}$, a direct component. We consider three confounding dimensions, $K=5,10,$ and $20$.

\subsection*{Active scale $j=1$}

For the design with active scale $j=1$, the within-block spectra for $\{\mathbf{X}_t\}$ and $\{\mathbf{Y}_t\}$ are
\[
\mathbf{S}_{1;\mathbf{XX}}=
\begin{bmatrix}
8 & 1 & 1 & 0\\
1 & 8 & 0 & 1\\
1 & 0 & 8 & 0\\
0 & 1 & 0 & 8
\end{bmatrix},
\qquad
\mathbf{S}_{1;\mathbf{YY}}=
\begin{bmatrix}
5 & 1 & 1\\
1 & 5 & 0\\
1 & 0 & 5
\end{bmatrix}.
\]
In the direct-dependence setting, we take
\[
\boldsymbol{\Delta}_{XY}^{(1)}(u)=
\begin{bmatrix}
0 & 0 & 1+u\\
1+u & 0 & 0\\
0 & 0 & 0\\
0 & 0 & 0
\end{bmatrix},
\]
whereas in the no-direct setting we set $\boldsymbol{\Delta}_{XY}^{(1)}(u)\equiv \mathbf{0}$.

For $K=5$, the confounding block is
\[
\mathbf{S}_{1;\mathbf{ZZ}}^{(K=5)}=
\begin{bmatrix}
8 & 0 & 2.5 & 0 & 0\\
0 & 8 & 2   & 1.5 & 1\\
2.5 & 2 & 8 & 0 & 0\\
0 & 1.5 & 0 & 8 & 0\\
0 & 1 & 0 & 0 & 8
\end{bmatrix},
\]
with
\[
\mathbf{S}_{1;\mathbf{XZ}}^{(K=5)}(u)=
\begin{bmatrix}
0 & 0 & 2+u & 0 & 0\\
2+u & 0 & 0 & 0 & 0\\
0 & 0 & 0 & 0 & 0\\
0 & 0 & 0 & 0 & 0
\end{bmatrix},
\qquad
\mathbf{S}_{1;\mathbf{YZ}}^{(K=5)}(u)=
\begin{bmatrix}
0 & 0 & 1+u & 0 & 0\\
1+u & 0 & 0 & 0 & 0\\
0 & 0 & 0 & 0 & 0
\end{bmatrix}.
\]

For $K=10$, five additional confounding channels are appended. These channels have diagonal entries equal to $7$, are linked sparsely to the original five confounders, and a subset of them is further connected to selected components of $\mathbf{X}_t$ or $\mathbf{Y}_t$ through weaker time-varying cross-spectra. Thus, the $K=10$ setting preserves the same core confounding structure while increasing the observed dimension of $\mathbf{Z}_t$.

For $K=20$, ten further confounding channels are appended to the $K=10$ specification. These additional channels have diagonal entries equal to $7$ and are incorporated through sparse links: some are connected to previously defined confounders, some are linked to selected components of $\mathbf{X}_t$ or $\mathbf{Y}_t$, and some remain only weakly connected. This yields a higher-dimensional and more heterogeneous confounding block while leaving $\mathbf{S}_{1;\mathbf{XX}}$, $\mathbf{S}_{1;\mathbf{YY}}$, and $\boldsymbol{\Delta}_{XY}^{(1)}(u)$ unchanged. The corresponding results appear in Section \ref{sec:sim} of the main text.

\subsection*{Active scale $j=4$}

For the design with active scale $j=4$, the within-block spectra become
\[
\mathbf{S}_{4;\mathbf{XX}}=
\begin{bmatrix}
8 & 1 & 0 & 1\\
1 & 8 & 1 & 0\\
0 & 1 & 8 & 0\\
1 & 0 & 0 & 8
\end{bmatrix},
\qquad
\mathbf{S}_{4;\mathbf{YY}}=
\begin{bmatrix}
5 & 1 & 0\\
1 & 5 & 1\\
0 & 1 & 5
\end{bmatrix}.
\]
In the direct-dependence setting, we take
\[
\boldsymbol{\Delta}_{XY}^{(4)}(u)=
\begin{bmatrix}
0 & 2-u & 0\\
0 & 0 & 0\\
2-u & 0 & 0\\
0 & 0 & 0
\end{bmatrix},
\]
whereas in the no-direct setting we set $\boldsymbol{\Delta}_{XY}^{(4)}(u)\equiv \mathbf{0}$.

For $K=5$, the core confounding block is
\[
\mathbf{S}_{4;\mathbf{ZZ}}^{(K=5)}=
\begin{bmatrix}
8 & 2 & 0 & 1.5 & 0\\
2 & 8 & 0 & 0 & 1\\
0 & 0 & 8 & 2.2 & 0\\
1.5 & 0 & 2.2 & 8 & 0\\
0 & 1 & 0 & 0 & 8
\end{bmatrix},
\]
with
\[
\mathbf{S}_{4;\mathbf{XZ}}^{(K=5)}(u)=
\begin{bmatrix}
0 & 3-0.8u & 0 & 0 & 0\\
0 & 0 & 0 & 0 & 0\\
0 & 0 & 0 & 2.6-0.8u & 0\\
0 & 0 & 0 & 0 & 0
\end{bmatrix},
\qquad
\mathbf{S}_{4;\mathbf{YZ}}^{(K=5)}(u)=
\begin{bmatrix}
0 & 2-0.8u & 0 & 0 & 0\\
0 & 0 & 0 & 0 & 0\\
0 & 0 & 0 & 1.8-0.8u & 0
\end{bmatrix}.
\]

As in the $j=1$ design, the cases $K=10$ and $K=20$ are obtained by augmenting the confounding block with additional channels having diagonal entries equal to $7$, together with sparse internal links and weaker time-varying connections to selected components of $\mathbf{X}_t$ and $\mathbf{Y}_t$. The direct component $\boldsymbol{\Delta}_{XY}^{(4)}(u)$ is kept fixed across these higher-dimensional settings. For numerical stability, we add a ridge term $10^{-6}\mathbf{I}$ to $\mathbf{S}_{j;\mathbf{W}\mathbf{W}}(u)$ at each $(j,u)$.

Figures \ref{fig1:app_comparision} and \ref{fig2:app_high_dim} illustrate the estimates for the $j=4$ design. The same qualitative conclusions as in the main text continue to hold. In the no-direct setting, the proposed partial WaveCanCoh remains close to zero, whereas the unadjusted WaveCanCoh reflects the marginal association induced by the confounders. In the direct-dependence setting, the proposed method continues to recover the underlying time-varying direct association and remains stable as the confounder dimension increases.

\begin{figure}[htbp]
  \centering
  \begin{minipage}[b]{0.49\textwidth}
    \centering
    \includegraphics[width=\linewidth,height=6cm,keepaspectratio]{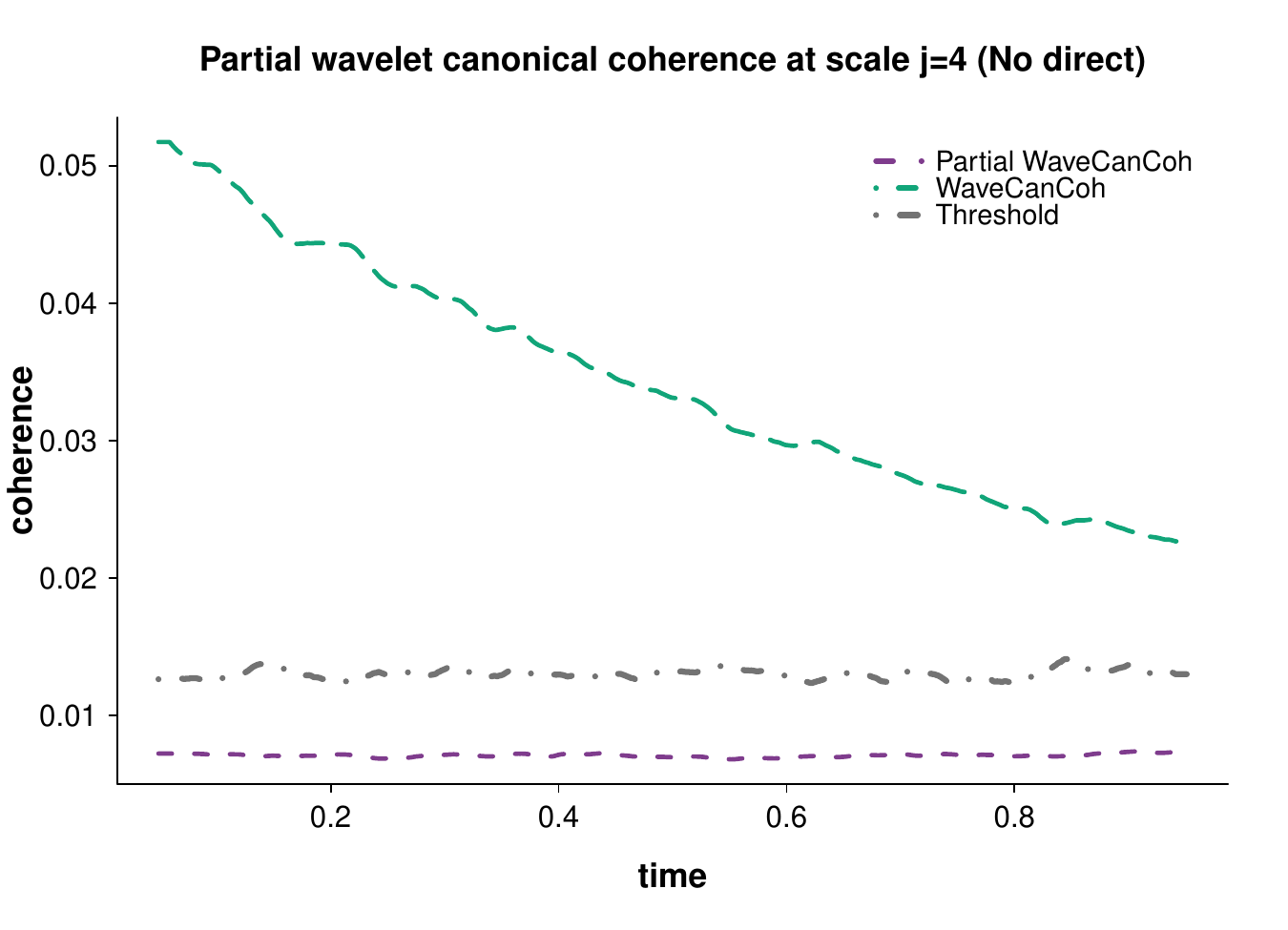}
  \end{minipage}
  \hfill
  \begin{minipage}[b]{0.49\textwidth}
    \centering
    \includegraphics[width=\linewidth,height=6cm,keepaspectratio]{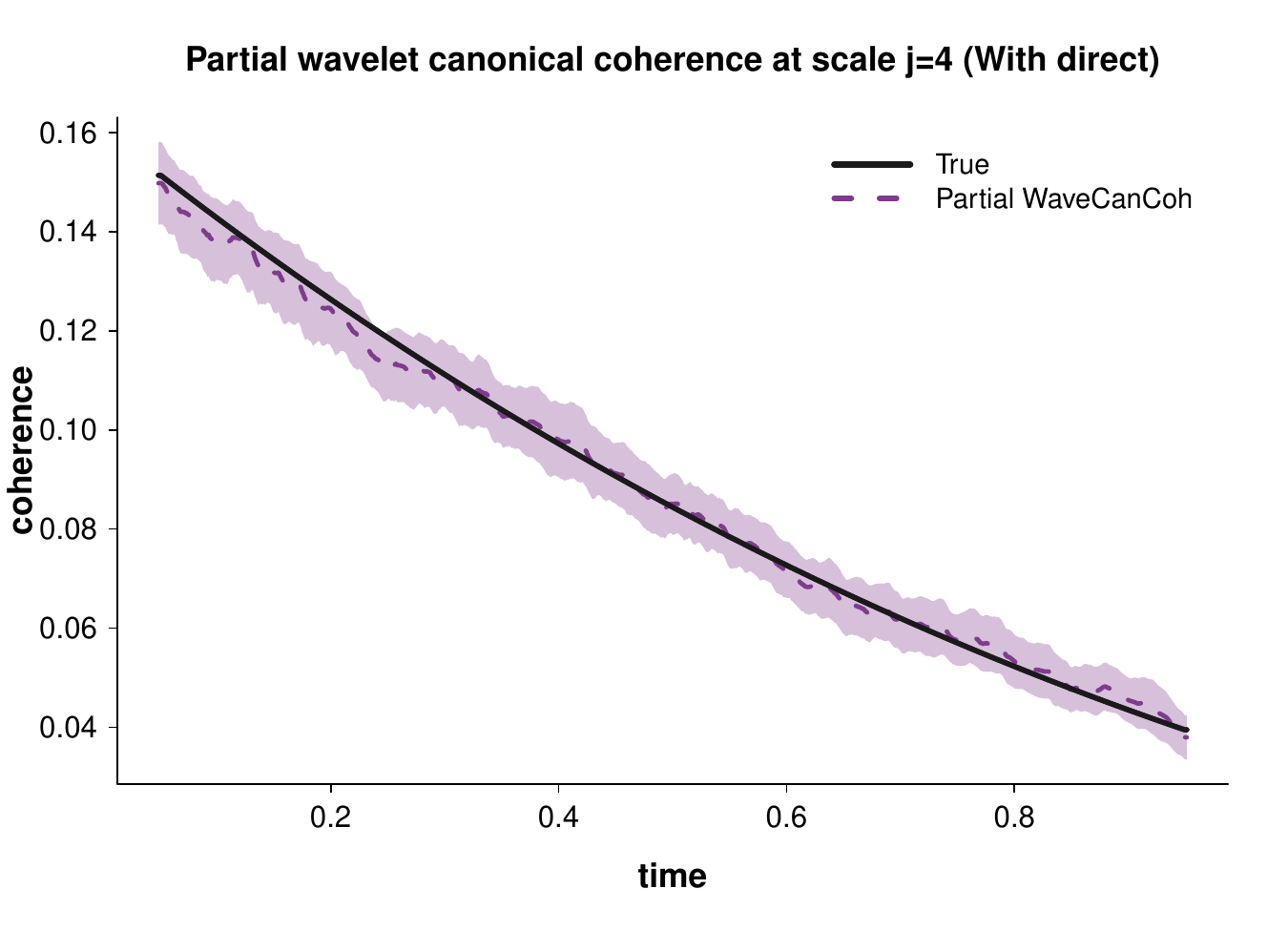}
  \end{minipage}
  \caption{Estimated partial wavelet canonical coherence (Partial WaveCanCoh) at scale $j=4$ under no direct dependence (left) and time-varying direct dependence (right). Estimated curves are Monte Carlo averages over 1000 independent replicates. The gray dashed curve in the left panel is the pointwise 95\% empirical null threshold. Shaded bands in the right panel denote Wald-type 95\% confidence intervals, with the true curve shown in black.}
  \label{fig1:app_comparision}
\end{figure}

\begin{figure}[htbp]
  \centering
  \includegraphics[width=0.6\linewidth]{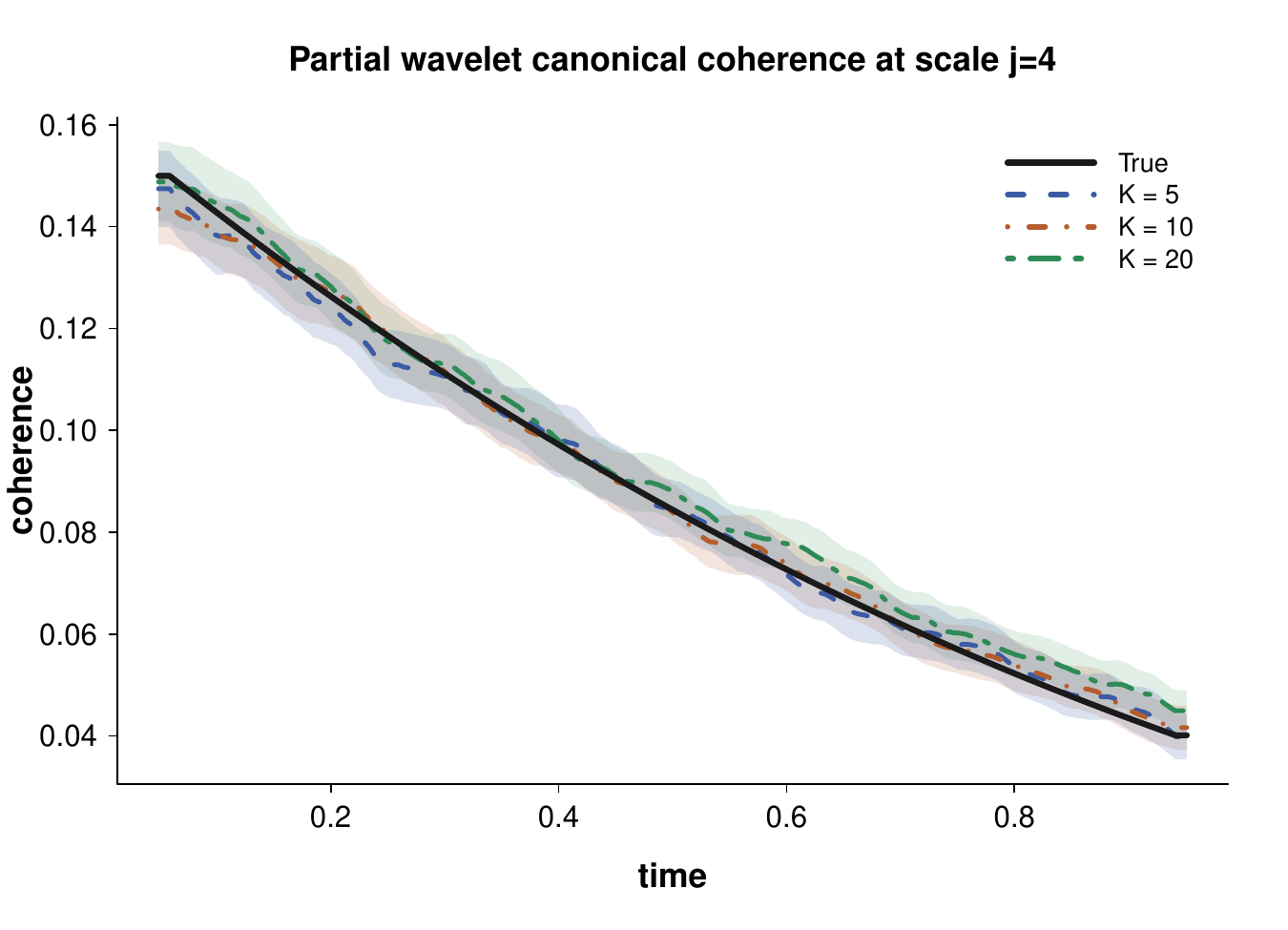}
  \caption{Partial wavelet canonical coherence at scale $j=4$ under high-dimensional confounding. The true curve is shown in black, and estimates are reported for $K=5,\, 10, \, 20$. Shaded bands denote Wald-type 95\% confidence intervals. The proposed method continues to recover the underlying time-varying direct association as the confounder dimension increases.}
  \label{fig2:app_high_dim}
\end{figure}

To further assess calibration under the no-direct-dependence setting, we report the pointwise empirical Type I error rates at scales $j=1$ and $j=4$ in Figure \ref{fig:type1_null}, illustrating the empirical rejection proportions remain close to the nominal 5\% level across time.
\begin{figure}[htbp]
  \centering
  
  \begin{minipage}[b]{0.49\textwidth}
    \centering
    \includegraphics[width=\linewidth,height=5.6cm,keepaspectratio]{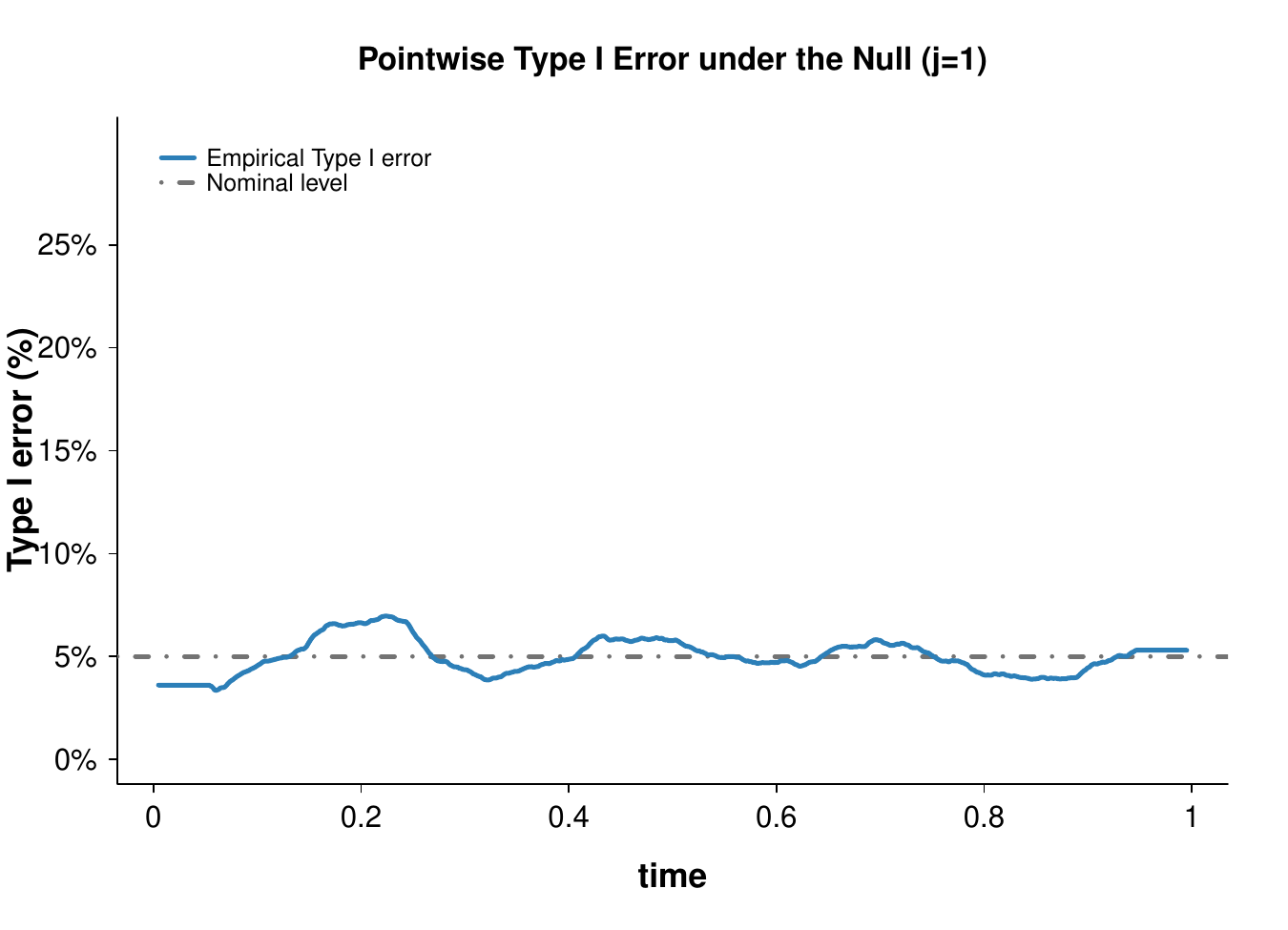}
    
    % \smallskip
    % {\small (a) $j=1$}
  \end{minipage}
  \hfill
  \begin{minipage}[b]{0.49\textwidth}
    \centering
    \includegraphics[width=\linewidth,height=5.6cm,keepaspectratio]{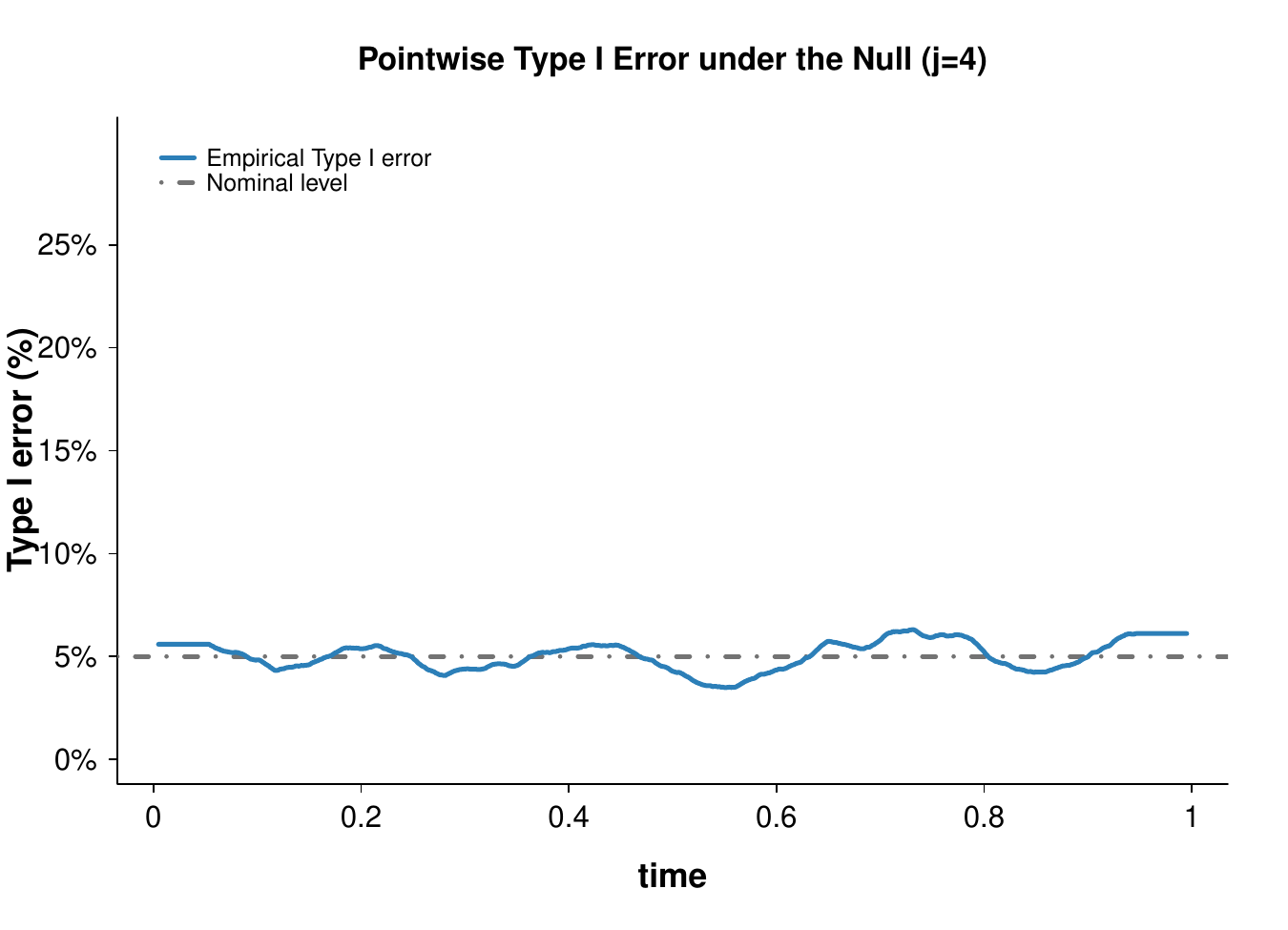}
    
    % \smallskip
    % {\small (b) $j=4$}
  \end{minipage}
  
  \caption{Pointwise empirical Type I error under the no direct dependence setting. The curves report the proportion of independent null evaluation replicates for which the estimated partial WaveCanCoh exceeds the pointwise 95\% empirical null threshold. The dashed horizontal line denotes the nominal 5\% level.}
  \label{fig:type1_null}
\end{figure}

\section{Additional information for the financial data analysis}\label{app:data_analysis}

This appendix reports the composition of the ETF groups used in the empirical analysis. 
The sets $\{\mathbf{X}_t\}$ and $\{\mathbf{Y}_t\}$ are kept fixed throughout, whereas the confounding block $\{\mathbf{Z}_t\}$ is constructed in a nested manner so that its dimension increases from 5 to 10, 20, and 30. 
This design preserves a coherent economic interpretation of $\{\mathbf{Z}_t\}$ as an external market block while allowing us to examine the proposed method under increasingly high-dimensional confounding settings.

The fixed groups are given by
\begin{align*}
\mathbf{X}_t
&= \bigl(r_{\mathrm{SPY},t},\, r_{\mathrm{QQQ},t},\, r_{\mathrm{IWM},t},\, r_{\mathrm{XLK},t}\bigr)^\top,\\
\mathbf{Y}_t
&= \bigl(r_{\mathrm{XLF},t},\, r_{\mathrm{XLI},t},\, r_{\mathrm{XLY},t},\, r_{\mathrm{XLE},t}\bigr)^\top,
\end{align*}
where $r_{a,t}$ denotes the daily log return of asset $a$ at time $t$.

The smallest confounding specification is
\begin{align*}
\mathbf{Z}_t^{(5)}
&=
\bigl(r_{\mathrm{EFA},t},\, r_{\mathrm{EEM},t},\, r_{\mathrm{EZU},t},\, r_{\mathrm{EWJ},t},\, r_{\mathrm{FXI},t}\bigr)^\top.
\end{align*}
The higher-dimensional specifications are obtained by progressively enlarging this set with additional U.S.-listed international equity ETFs:
\begin{align*}
\mathbf{Z}_t^{(10)} &\supset \mathbf{Z}_t^{(5)},\\
\mathbf{Z}_t^{(20)} &\supset \mathbf{Z}_t^{(10)},\\
\mathbf{Z}_t^{(30)} &\supset \mathbf{Z}_t^{(20)}.
\end{align*}
In the higher-dimensional settings, PCA is applied to the resulting $\mathbf{Z}_t$ block prior to the Partial WaveCanCoh analysis.

\begin{table}[htbp]
\centering
\caption{Fixed ETF groups used for $\mathbf{X}_t$ and $\mathbf{Y}_t$.}
\label{tab:xy_groups}
\begin{tabular}{lll}
\toprule
Group & Ticker & Economic interpretation \\
\midrule
$\mathbf{X}_t$ & SPY & Broad U.S. large-cap market exposure \\
$\mathbf{X}_t$ & QQQ & U.S. growth / Nasdaq-100 exposure \\
$\mathbf{X}_t$ & IWM & U.S. small-cap exposure \\
$\mathbf{X}_t$ & XLK & U.S. technology sector exposure \\
\midrule
$\mathbf{Y}_t$ & XLF & U.S. financial sector exposure \\
$\mathbf{Y}_t$ & XLI & U.S. industrial sector exposure \\
$\mathbf{Y}_t$ & XLY & U.S. consumer discretionary sector exposure \\
$\mathbf{Y}_t$ & XLE & U.S. energy sector exposure \\
\bottomrule
\end{tabular}
\end{table}
\begin{center}
\footnotesize
\setlength{\tabcolsep}{4pt}
\renewcommand{\arraystretch}{0.95}

\begin{longtable}{p{0.24\textwidth} p{0.10\textwidth} p{0.56\textwidth}}
\caption{Nested construction of the confounding block $\mathbf{Z}_t$.}
\label{tab:z_groups}\\
\toprule
Expansion step & Ticker & Economic interpretation \\
\midrule
\endfirsthead

\toprule
Expansion step & Ticker & Economic interpretation \\
\midrule
\endhead

\midrule
\multicolumn{3}{r}{\textit{Continued on next page}} \\
\endfoot

\bottomrule
\endlastfoot

Core set for $\mathbf{Z}_t^{(5)}$ & EFA  & Developed markets ex U.S. and Canada \\
Core set for $\mathbf{Z}_t^{(5)}$ & EEM  & Broad emerging markets exposure \\
Core set for $\mathbf{Z}_t^{(5)}$ & EZU  & Eurozone equity market exposure \\
Core set for $\mathbf{Z}_t^{(5)}$ & EWJ  & Japan equity market exposure \\
Core set for $\mathbf{Z}_t^{(5)}$ & FXI  & China large-cap equity exposure \\

\midrule
Add to obtain $\mathbf{Z}_t^{(10)}$ & EWU  & United Kingdom equity market exposure \\
Add to obtain $\mathbf{Z}_t^{(10)}$ & EWC  & Canada equity market exposure \\
Add to obtain $\mathbf{Z}_t^{(10)}$ & EWA  & Australia equity market exposure \\
Add to obtain $\mathbf{Z}_t^{(10)}$ & INDA & India equity market exposure \\
Add to obtain $\mathbf{Z}_t^{(10)}$ & EWT  & Taiwan equity market exposure \\

\midrule
Add to obtain $\mathbf{Z}_t^{(20)}$ & EWG  & Germany equity market exposure \\
Add to obtain $\mathbf{Z}_t^{(20)}$ & EWQ  & France equity market exposure \\
Add to obtain $\mathbf{Z}_t^{(20)}$ & EWI  & Italy equity market exposure \\
Add to obtain $\mathbf{Z}_t^{(20)}$ & EWL  & Switzerland equity market exposure \\
Add to obtain $\mathbf{Z}_t^{(20)}$ & EWP  & Spain equity market exposure \\
Add to obtain $\mathbf{Z}_t^{(20)}$ & EWD  & Sweden equity market exposure \\
Add to obtain $\mathbf{Z}_t^{(20)}$ & EWK  & Belgium equity market exposure \\
Add to obtain $\mathbf{Z}_t^{(20)}$ & EWY  & South Korea equity market exposure \\
Add to obtain $\mathbf{Z}_t^{(20)}$ & EWH  & Hong Kong equity market exposure \\
Add to obtain $\mathbf{Z}_t^{(20)}$ & EWS  & Singapore equity market exposure \\

\midrule
Add to obtain $\mathbf{Z}_t^{(30)}$ & EWZ  & Brazil equity market exposure \\
Add to obtain $\mathbf{Z}_t^{(30)}$ & EWW  & Mexico equity market exposure \\
Add to obtain $\mathbf{Z}_t^{(30)}$ & EPU  & Peru equity market exposure \\
Add to obtain $\mathbf{Z}_t^{(30)}$ & ECH  & Chile equity market exposure \\
Add to obtain $\mathbf{Z}_t^{(30)}$ & ARGT & Argentina equity market exposure \\
Add to obtain $\mathbf{Z}_t^{(30)}$ & EPHE & Philippines equity market exposure \\
Add to obtain $\mathbf{Z}_t^{(30)}$ & THD  & Thailand equity market exposure \\
Add to obtain $\mathbf{Z}_t^{(30)}$ & ENZL & New Zealand equity market exposure \\
Add to obtain $\mathbf{Z}_t^{(30)}$ & TUR  & Turkey equity market exposure \\
Add to obtain $\mathbf{Z}_t^{(30)}$ & EIDO & Indonesia equity market exposure \\
\end{longtable}
\end{center}

\vspace{-0.5em}
\noindent\footnotesize
\textit{Notes.}
The construction is nested: $\mathbf{Z}_t^{(10)}$ contains all ETFs in $\mathbf{Z}_t^{(5)}$ plus the five ETFs added at the next step; $\mathbf{Z}_t^{(20)}$ contains all ETFs in $\mathbf{Z}_t^{(10)}$ plus the ten ETFs added for the 20-dimensional specification; and $\mathbf{Z}_t^{(30)}$ contains all ETFs in $\mathbf{Z}_t^{(20)}$ plus the ten ETFs added for the 30-dimensional specification. In the 20- and 30-dimensional settings, the confounding block is reduced by PCA, in accordance with Algorithm~\ref{alg:alg_partial_hd}. %the partial WaveCanCoh procedure.  
\end{appendices}

\end{document}